\documentclass[11pt]{article}

\usepackage[a4paper,margin=1in]{geometry}
\usepackage{amsmath,amssymb,amsthm}
\usepackage{mathtools}
\usepackage[
  hidelinks,
  linktoc=all,
  bookmarksopen=true,
  bookmarksnumbered=true
]{hyperref}
\usepackage{enumitem}

\newtheorem{theorem}{Theorem}
\newtheorem{lemma}{Lemma}
\newtheorem{proposition}{Proposition}
\newtheorem{definition}{Definition}
\newtheorem{remark}{Remark}

\newcommand{\C}{\mathcal{C}}

\newcommand{\NC}{\mathrm{NC}}
\newcommand{\NCHV}{\mathrm{NCHV}}
\newcommand{\GNCHV}{\mathrm{GNCHV}}
\newcommand{\GLHV}{\mathrm{GLHV}}
\newcommand{\Cost}{\mathsf{Cost}}
\newcommand{\TV}{\mathsf{TV}}

\title{Genuine Global Kochen--Specker Contextuality and Classical Coordination Costs}
\author{Ming Yang}
\date{\today}

\begin{document}
\maketitle

\begin{abstract}
We formulate classical simulation as coordination across a spacetime
separator.  If \(B_\Sigma\) classical message bits and \(M_\Sigma\) persistent
classical-state bits carry all past-dependent information across a cut
\(\Sigma\), then every classical causal simulation of a correlation table
\(P\) satisfies
\[
 B_\Sigma+M_\Sigma
 \geq \log_2\operatorname{rank}_+F_\Sigma(P),
\]
with approximate and amortized forms given by approximate nonnegative rank
and Wyner common information.  For sequential processes we introduce a causal
positive-realization rank whose factors must share consistent response and
update maps.  Local computation depth \(D\) restricts those maps; it is not an
additional bit count.  Support covering numbers, including
\(B+M\geq\log_2\chi_G^D\), then appear as restricted chart-model projections
of this general operational hierarchy.

We apply the framework to genuine global Kochen--Specker contextuality:
local subsystems are noncontextual and the tested multipartite blocks are
generalized-Bell-local, while the complete empirical model has no global
noncontextual hidden-variable explanation.  For the polarization--path Hardy
table of Ref.~\cite{yang2026global}, at its symmetric point the
Alice--Bob event matrix has
\(\operatorname{rank}_+=4\) and positive-semidefinite rank \(2\).
Thus one exact classical separator needs two coordination bits whereas the
quantum separator is one qubit.  Exact repetition gives the classical lower
bound \(m\log_2 3\); asymptotically vanishing-error generation requires
\(\frac32m-o(m)\) classical bits, compared with \(m\) qubits.  The earlier
seven-test Hardy rate \(m\log_2(25/24)\) is retained as a possibilistic
cover-model comparison.

Finally, we construct a finite family whose empirical restriction is
genuinely global contextual and whose full adaptive behavior exhibits an
exact quadratic separation.  A finite stabilizer witness derived from the
Karanjai--Wallman--Bartlett counting argument is combined with a
Peres--Mermin branch.  An explicit mixture over the \(96\) assignments that
violate exactly one Peres--Mermin context gives a size-independent flag
visibility \(p=2/3\).  Conditioning preserves the classical boundary-state
set and yields
\[
 B_{\Sigma_n}+M_{\Sigma_n}
 \geq
 \log_2\!\frac{2^n\prod_{j=1}^n(2^j+1)}
 {5\,3^{n-2}}
 =
 \frac12n^2+
 \left(\frac32-\log_2 3\right)n+O(1),
\]
whereas the quantum boundary uses at most \(n+1\) qubits.  This is an exact
separation against arbitrary finite-state causal online simulators.  A
constant-error quadratic lower bound and a flag-free intrinsic construction
remain open.
\end{abstract}

\clearpage
\setcounter{tocdepth}{2}
\tableofcontents
\clearpage

\section{Introduction}

The starting point of this note is that classical simulation is a problem of
coordinating information across spacetime events.  Communication, memory, and
local computation are usually introduced as distinct resources, but this
distinction depends on how one cuts the same causal information flow.  A
message sent between laboratories is information crossing a spatial boundary;
a memory register is information carried from one event to a later event; a
local computation is information processed inside a bounded causal region.

A Turing machine gives a simple discrete illustration.  Consider a transition
\[
  \delta(q,s)=(q',s',\mathrm{move}).
\]
The usual reading is that the machine has performed one elementary
computational step.  But the same step may be read as a local communication
process: the tape cell sends the current symbol \(s\) to the finite controller;
the controller sends back a new symbol \(s'\); and the controller sends a
movement command to the head.  Thus a Turing step is a local exchange of
information between a controller and a storage medium.

Under this reading,
\[
  \text{computation} = \text{internal information transfer plus state update},
\]
while
\[
  \text{communication} = \text{information transfer across a chosen boundary}.
\]
Memory is the local snapshot retained after information transfer.  Time,
space, and communication complexity are therefore not unrelated quantities.
They are different projections of one question: how much coordination is
required to maintain and update a global information state over a causal
network of events?

Quantum foundations already provide two mature instances of this question.
Bell nonlocality shows that some multipartite correlations cannot be simulated
by zero-communication local hidden-variable models.  Classical simulation must
then pay communication.  Kochen--Specker contextuality shows that some
measurement statistics cannot be simulated by a single noncontextual hidden
state.  Classical simulation must then pay memory, hidden state, or state
update complexity.

This note asks whether genuine global Kochen--Specker contextuality gives the
next member of the same pattern.  Here ``genuine global'' is used in the
physical, stratified sense: local subsystems and the observed multipartite
blocks can each be classically explained, while the empirical model of the
whole physical system cannot be explained by one global noncontextual
hidden-variable model.  A formal statement is given in
Sec.~\ref{sec:genuine-global}.  We ask how this residual obstruction is
reflected in the cost of a general classical causal simulator.  This requires
separating two questions that are often conflated.  A support-cover argument
counts context-independent global charts and applies to simulators that select
among such charts.  A general classical simulator may instead use randomized,
context-dependent responses and adaptive state updates.  Its universal cost
is controlled by the classical information crossing a spacetime separator,
represented below by nonnegative factorizations and causal positive
realizations.  Chart covering is therefore a useful restricted layer, not the
definition of general classical simulation.

\paragraph{Informal thesis.}
\[
\begin{array}{rcl}
\text{Bell nonlocality} &\Longrightarrow& \text{communication cost},\\[1mm]
\text{KS contextuality} &\Longrightarrow& \text{memory/space cost},\\[1mm]
\text{genuine global KS contextuality} &\Longrightarrow&
\text{joint classical coordination cost}.
\end{array}
\]
The arrows state operational research programs, not one-to-one identities.
Bell and KS obstructions identify structures that can force classical
communication or memory, but the numerical lower bound depends on the
simulator class and the selected causal cut.  Genuine global KS contextuality
is particularly useful because its local and blockwise classicality prevents
the nonclassicality from being assigned to one subsystem or one Bell block;
the remaining question is how much information must coordinate the complete
spacetime process.  The common accounting unit is a bit of past-dependent
classical information crossing or persisting through the cut.  Computation
depth has a different unit and restricts which response and update maps are
available.

\paragraph{Organization.}
Section~2 reviews contextuality, communication, memory, and factorization
approaches.  Section~3 defines spacetime separators and proves the general
classical and quantum realization bounds, before relating them to Turing
machines.  Sections~4--7 introduce empirical tasks and the restricted
chart-cover hierarchy, including depth, rectangle, repetition, and Tseitin
calculations.  Section~8 recovers the familiar Bell and KS readings.
Sections~9--11 analyze genuine global contextuality, the full Hardy table, its
seven-test support projection, and the postselected KCBS comparison.
Section~12 gives the finite stabilizer flag construction and the exact
quadratic coordination separation.  The final section states the operational
hierarchy, limitations, and open robust problems.

\section{Related Work}

Bell nonlocality and communication complexity have a long connection, starting
from the use of entanglement to reduce communication in distributed tasks and
from results showing that stronger-than-quantum nonlocal boxes can trivialize
communication complexity.  See, for example, Buhrman, Cleve, Massar, and de
Wolf~\cite{buhrman2010nonlocality}, Brassard et
al.~\cite{brassard2006limit}, and later graph- and inequality-based
formulations.

The general bipartite generation problem is described by classical
correlation complexity and nonnegative rank, with the quantum counterpart
controlled by positive-semidefinite rank~\cite{jain2012correlation}.  Its
asymptotic common-randomness rate is Wyner common
information~\cite{wyner1975common}.  For temporal processes, the analogous
classical problem is positive realization by a hidden-state model, while
quantum memory leads to completely positive realizations by quantum
instruments~\cite{monras2016quantum}.  These are the immediate mathematical
antecedents of the separator hierarchy used here.

The memory cost of contextuality has been studied from a complementary angle:
instead of asking how much information must cross a spatial boundary, one asks
how much internal state is needed for a classical machine to reproduce
contextual correlations across measurement contexts.  See Kleinmann et
al.~\cite{kleinmann2011memory}, Fagundes and
Kleinmann~\cite{fagundes2017memory}, and Karanjai, Wallman, and
Bartlett~\cite{karanjai2018contextuality}.

The mathematical unity of nonlocality and contextuality is expressed by the
sheaf-theoretic framework of Abramsky and
Brandenburger~\cite{abramsky2011sheaf}.  Quantitative versions include the
contextual fraction~\cite{abramsky2017contextual}.  Trade-offs between Bell
nonlocality and local KS contextuality in generalized Bell scenarios were
studied in Ref.~\cite{porto2024tradeoff}.  Yang subsequently exhibited a
distinct residual case: models that are neither locally contextual nor
generalized-Bell-nonlocal, but still fail to admit a global noncontextual
explanation~\cite{yang2026global}.

There are two notions of ``global'' in play.  In the sheaf-theoretic sense,
globality refers to a section over all measurements in a measurement cover,
while locality refers to individual compatible contexts.  In genuinely global
KS contextuality, globality refers instead to a physical stratification of the
system: local subsystems and observed multipartite blocks may each admit
classical explanations, while the total empirical model over the full physical
system does not.  The sheaf-theoretic framework supplies the local-to-global
language; the present paper asks how much classical communication, memory, or
local computation is needed to coordinate around such obstructions.

To move beyond specially designed games, we also use the standard view of
constraint satisfaction problems and relational structures as local-to-global
consistency problems.  A useful related formulation is the database reading of
local consistency and global joins~\cite{abramsky2012databases}.  Measurement-based
quantum computation provides a complementary route from general quantum
algorithms to local measurement patterns with classical
feed-forward~\cite{raussendorf2003mbqc,briegel2009mbqc,raussendorf2013contextual}.
The computational role of contextuality and magic is studied, for example, in
Refs.~\cite{howard2014magic,karanjai2018contextuality}.  The depth axis is
also motivated by shallow-circuit separations between quantum and classical
models~\cite{bravyi2018shallow}.

\section{Coordination Across Spacetime Boundaries}
\label{sec:separator}

The most physical version of the preceding thesis is not tied to Turing
machines.  It starts from spacetime events.  Let \(\mathcal E\) be a finite
set of measurement or computation events embedded in a relativistic spacetime,
and let
\[
  e\preceq e'
\]
denote the causal order: information can be sent from \(e\) to \(e'\) only if
\(e'\) lies in the future light cone of \(e\).  A local experimental context is
a finite set
\[
  C\subseteq \mathcal E
\]
whose outcomes are jointly registered in one operational run.  A deterministic
global classical chart is a context-independent assignment
\[
  g:\mathcal E\to O,
\]
assigning an outcome or local state to every relevant spacetime event.

In this language the three resources are bookkeeping choices for the same
causal structure:
\[
\begin{array}{rcl}
B &:& \text{information crossing a chosen spacelike or inter-laboratory boundary},\\[1mm]
M &:& \text{information carried forward along timelike worldlines},\\[1mm]
D &:& \text{local processing depth inside causal diamonds}.
\end{array}
\]
Thus communication is coordination across spatial separation, memory is
coordination across time, and local computation is coordination internal to a
bounded spacetime region.  These are not separate primitives from the
relativistic viewpoint; they are different cuts through the same causal
information flow.

This spacetime language is not a new definition of Bell nonlocality or
Kochen--Specker contextuality.  Bell nonlocality remains the failure of a
local hidden-variable factorization, and KS contextuality remains the failure
of a context-independent valuation or global section.  The point is
operational: when the same local-to-global obstruction is embedded in a
relativistic experiment or simulator, different choices of spacetime cut make
the required classical coordination appear as communication, memory, local
processing, or a mixture of them.  A Bell experiment emphasizes spacelike
separation because the relevant events cannot exchange messages during one
run.  A KS experiment emphasizes compatibility contexts; when such contexts
are simulated sequentially, the same obstruction may appear as hidden-state or
memory update cost.  Genuine global KS contextuality is the case in which no
single standard cut accounts for the obstruction: every local region and every
tested multipartite block may admit a classical explanation, but the local
explanations cannot be glued into a single chart on the union of all relevant
spacetime events.

Formally, for a spacetime cover \(\{U_\alpha\}\) of the event set
\(\mathcal E\), local classical explanations are charts
\[
  g_\alpha:U_\alpha\to O.
\]
Local compatibility requires agreement on overlaps,
\[
  g_\alpha|_{U_\alpha\cap U_\beta}
  =
  g_\beta|_{U_\alpha\cap U_\beta},
\]
or the corresponding probabilistic marginal consistency condition.  A global
noncontextual explanation is a single chart \(g:\mathcal E\to O\) whose
restriction to each \(U_\alpha\) gives the local explanation.  Global
contextuality is precisely the failure of this gluing problem.  The
coordination cost studied below measures how many classical bits must be sent,
stored, or generated locally to repair this failure by switching among several
global charts.

\subsection{Operational separators and resource units}

Fix a finite protocol directed acyclic graph and a separator \(\Sigma\) such
that every causal path from a designated past region to a designated future
region crosses \(\Sigma\).  The future receives a query \(x\) only after the
past-dependent boundary configuration has been formed.  Unless stated
otherwise, \(B_\Sigma\) and \(M_\Sigma\) are worst-case fixed-length
capacities, all shared randomness correlated with the past is counted, and
future test choices are exogenous and independent of the past.  We use the
following convention throughout:
\[
\begin{array}{rcl}
B_\Sigma&:&\text{classical message or transcript bits crossing }\Sigma,\\
M_\Sigma&:&\text{persistent classical-state bits retained across }\Sigma,\\
Q_\Sigma&:&\log_2\text{ of the quantum boundary factor dimension},\\
D&:&\text{depth of response and update maps in a specified local model}.
\end{array}
\]
For bipartite correlation generation, \(Q_\Sigma\) is the logarithm of the
Schmidt or positive-semidefinite factor dimension, equivalently half the
number of qubits in a minimal pure shared seed under the standard convention
of Ref.~\cite{jain2012correlation}.  For a temporal cut it is simply the
number of qubits retained across the cut.

All past-dependent information available in the future is counted.  In
particular, a saved preparation label, a rereadable public transcript, retained
shared randomness, and a past-dependent instruction contribute to
\(B_\Sigma+M_\Sigma\).  Fixed program code and a clock label are free.  Fresh
private randomness generated strictly to the future of \(\Sigma\) is also
free, but it is independent of the past and cannot enlarge the set of
past-dependent boundary configurations.  Computation depth is not added to a
bit count; it restricts which boundary response and update maps are available.

A finite-state classical causal realization has a boundary state
\(\lambda\in\Lambda\), response kernels
\[
  \xi_x(a\mid\lambda),
\]
and stochastic update kernels
\[
  \Gamma_x(\lambda',a\mid\lambda).
\]
The kernels may depend on the next operation \(x\), but not for free on the
complete realized history.  If that history can be inspected later, its
future-accessible representation is part of the counted boundary
configuration.  Consequently,
\[
  |\Lambda|\leq 2^{B_\Sigma+M_\Sigma}.
\]

\subsection{The general classical separator bound}

Choose finite alphabets \(\mathcal U_\Sigma\) and \(\mathcal V_\Sigma\) of
past and future events, including any test inputs needed to specify the
experiment, and fix their test distribution.  The separator matrix is the
normalized nonnegative table
\[
  F_\Sigma(P)_{u,v}=\Pr_P(U=u,V=v).
\]
Future exogenous test choices are taken independent of the past, or the
factorization is read conditionally on those choices.
Its nonnegative rank is the smallest \(K\) for which
\[
  F_\Sigma(P)_{u,v}
  =
  \sum_{\lambda=1}^{K}A_{u,\lambda}B_{\lambda,v},
  \qquad A_{u,\lambda},B_{\lambda,v}\geq0.
\]

\begin{theorem}[Classical spacetime-separator bound]
\label{thm:general-separator}
Let a classical causal protocol realize \(P\) on a finite protocol DAG.
Suppose that \(\lambda\) is the complete future-accessible configuration on
\(\Sigma\): after conditioning on \(\lambda\) and on any exogenous future test
choices, no additional past-dependent variable is available to the future.
Then
\[
  B_\Sigma+M_\Sigma
  \geq
  \log_2\operatorname{rank}_+F_\Sigma(P).
\]
This statement permits randomized, context-dependent, and adaptive classical
responses.
\end{theorem}

\begin{proof}
Completeness of the separator and independence of future test choices give
the Markov relation
\[
  U-\lambda-V.
\]
Conditioning on \(\lambda\) therefore factors the past and future event
weights, producing a nonnegative factorization of \(F_\Sigma(P)\) with at
most \(|\Lambda|\) terms.  Hence
\[
  \operatorname{rank}_+F_\Sigma(P)
  \leq|\Lambda|
  \leq2^{B_\Sigma+M_\Sigma}.
\]
Taking logarithms proves the claim.
\end{proof}

For total-variation error \(\epsilon\), define
\[
 \operatorname{rank}_+^\epsilon(F)
 =
 \min\left\{
 \operatorname{rank}_+(\widetilde F):
 \widetilde F\text{ is a normalized nonnegative table and }
 \TV(\widetilde F,F)\leq\epsilon
 \right\}.
\]
The same proof gives
\[
 B_\Sigma+M_\Sigma
 \geq\log_2\operatorname{rank}_+^\epsilon(F_\Sigma(P)).
\]
These relations are the one-shot correlation-complexity reading of
nonnegative rank~\cite{jain2012correlation}.

\subsection{Quantum factors and amortized information}

The positive-semidefinite rank of \(F\) is the least \(r\) such that
\[
 F_{u,v}=\operatorname{Tr}(A_uB_v)
\]
for \(r\times r\) positive-semidefinite matrices \(A_u,B_v\).  A quantum
boundary factor of dimension \(2^{Q_\Sigma}\) supplies such a factorization,
and exact bipartite correlation complexity satisfies
\[
 Q_\Sigma
 =
 \left\lceil
 \log_2\operatorname{rank}_{\rm psd}F_\Sigma(P)
 \right\rceil
\]
under the standard shared-seed convention~\cite{jain2012correlation}.
Thus a separation between nonnegative and positive-semidefinite rank is a
separation between classical and quantum coordination across the same cut.

There is also an entropy-rate form.  For the normalized pair \((U,V)\), let
\[
 C_{\rm W}(U;V)
 =
 \min_{\substack{U-Z-V}}
 I(U,V;Z)
\]
be Wyner's common information~\cite{wyner1975common}.  Any complete counted
classical boundary state \(Z\) with \(U-Z-V\) satisfies
\[
 B_\Sigma+M_\Sigma
 \geq H(Z)
 \geq I(U,V;Z)
 \geq C_{\rm W}(U;V).
\]
For fixed finite alphabets, \(m\) independent copies, unlimited fresh private
randomness to the future of the separator, and asymptotically vanishing
total-variation error, Wyner's distributed-generation theorem gives
\[
  B_m+M_m\geq m C_{\rm W}(U;V)-o(m).
\]
The common information is a lower-bound functional of the same classical
boundary resource, not a fourth resource axis.

\subsection{Sequential processes and causal realization rank}

For an online process, let \(R\) denote a tested past record and \(F\) a
permitted future adaptive event.  Its process Hankel table is
\[
  H(R,F)=\Pr(F\mid R).
\]
An arbitrary nonnegative factorization of \(H\) may assign unrelated
black-box responses to different futures.  To state the additional causal
condition explicitly, a \(K\)-state positive realization consists of
probability row vectors \(\alpha_R\in\mathbb R_+^K\) and substochastic
transition-emission matrices
\[
 T_{x,a}(\lambda,\lambda')
 =
 \Pr(a,\lambda'\mid x,\lambda),
 \qquad
 \sum_a T_{x,a}\mathbf 1=\mathbf 1.
\]
For every positive-probability extension \(R(x,a)\), the same matrices must
satisfy the Bayesian shift relation
\[
 \alpha_{R(x,a)}
 =
 \frac{\alpha_R T_{x,a}}
      {\alpha_R T_{x,a}\mathbf 1}.
\]
Each future event has a nonnegative continuation vector \(h_F\) such that
\[
 H(R,F)=\alpha_Rh_F,
 \qquad
 h_{(x,a)F}=T_{x,a}h_F.
\]
These common matrices, rather than a separate factor for every column, encode
one response-and-update mechanism for the whole process.  Define
\[
  \operatorname{rank}_+^{\rm causal}(H)
\]
as the minimum \(K\) in such a realization, and define
\(\operatorname{rank}_+^{{\rm causal},D}(H)\) by restricting all response and
update maps \((x,\lambda)\mapsto(a,\lambda')\) to the selected depth-\(D\)
circuit class.  Then
\[
 \operatorname{rank}_+(H)
 \leq
 \operatorname{rank}_+^{\rm causal}(H)
 \leq
 \operatorname{rank}_+^{{\rm causal},D}(H),
\]
and every depth-\(D\) classical online simulator obeys
\begin{equation}
\label{eq:causal-rank-bound}
 B_\Sigma+M_\Sigma
 \geq
 \log_2\operatorname{rank}_+^{{\rm causal},D}(H).
\end{equation}
Indeed, forgetting the shared shift equations gives an ordinary nonnegative
factorization, restricting the update maps can only increase the minimum
state count, and the boundary states of any actual simulator provide a
feasible causal realization with at most \(2^{B_\Sigma+M_\Sigma}\) states.
The right-hand side is nonincreasing with \(D\).  This is the process-level
\((B,M,D)\) tradeoff and the natural setting for automata, hidden Markov
models, Turing-machine updates, and the stabilizer-memory theorem used below.
The corresponding quantum object is a completely positive realization by a
quantum instrument~\cite{monras2016quantum}.

\subsection{Turing machines as discrete spacetime coordination}

The Turing-machine discussion is a discrete model of the same spacetime
principle.  We now make that interpretation slightly more formal.  Let
\[
  \mathcal M=(Q,\Gamma,\delta)
\]
be a deterministic single-tape Turing machine, with finite control states
\(Q\), tape alphabet \(\Gamma\), and transition map
\[
  \delta(q,s)=(q',s',d),
  \qquad d\in\{-1,0,+1\}.
\]
At time \(t\), write the configuration as
\[
  c_t=(q_t,h_t,\tau_t),
\]
where \(q_t\in Q\) is the control state, \(h_t\in\mathbb Z\) is the head
position, and \(\tau_t:\mathbb Z\to\Gamma\) is the tape contents.  Let
\[
  s_t=\tau_t(h_t)
\]
be the symbol at the active tape cell.  The global update is local:
\[
  (q_t,h_t,\tau_t(h_t))
  \longmapsto
  (q_{t+1},h_{t+1},\tau_{t+1}(h_t)),
\]
while every inactive cell \(i\neq h_t\) keeps \(\tau_{t+1}(i)=\tau_t(i)\).

This local update can be represented as a two-component exchange across the
active boundary between the finite controller and the active tape cell:
\[
\begin{array}{rcl}
\text{cell}\to\text{controller} &:& s_t,\\
\text{controller}\to\text{cell} &:& s'_t,\\
\text{controller}\to\text{head token} &:& d.
\end{array}
\]
Thus one Turing step has an internal boundary transcript of length at most
\[
  b_{\mathcal M}
  =
  2\lceil\log_2|\Gamma|\rceil+\lceil\log_2 3\rceil
\]
bits under this explicit read-write-move convention.  Other conventions shift
this constant, but not the point: a sequential computation is a repeated local
communication process plus state update.

\begin{proposition}[Boundary-communication simulation of a Turing run]
If \(\mathcal M\) runs for \(T\) steps and visits at most \(S\) tape cells, then
its run can be represented as a local communication process with internal
boundary transcript length at most
\[
  B_{\rm int}\leq T b_{\mathcal M},
\]
with stored configuration size at most
\[
  M_{\rm conf}
  \leq
  \lceil\log_2|Q|\rceil+\lceil\log_2 S\rceil
  +S\lceil\log_2|\Gamma|\rceil
\]
bits, up to the usual choice of a finite window containing the visited cells.
If the transition rule \(\delta\) is implemented by local circuits of depth
\(d_\delta\), the local computation depth of the run is \(D=T d_\delta\).
\end{proposition}

\begin{proof}
At each step only the finite control, the active cell, and the head token are
updated.  The active cell sends one alphabet symbol to the controller; the
controller returns one alphabet symbol and one movement command.  Summing this
constant-size exchange over \(T\) steps gives \(B_{\rm int}\leq T b_{\mathcal M}\).
The final configuration is determined by the control state, the head position
inside the visited window, and the \(S\) symbols in that window, giving the
stated storage bound.  The depth statement follows by composing the local
implementation of \(\delta\) for \(T\) rounds.
\end{proof}

The relevance for this paper is not the constant \(b_{\mathcal M}\).  The
point is that the same process can be cut in different ways.  If the boundary
is drawn between two laboratories, the transcript contributes to communication
cost.  If the boundary is drawn between a controller and its retained state,
the same information appears as memory.  Recomputing a transcript can trade
stored workspace for local depth only when the input label or earlier record
needed for that recomputation still crosses or persists through the cut.
Computation cannot recreate past-dependent information that is absent from
the future side of the separator.

This gives a technical version of the informal \((B,M,D)\) picture.  For a
compiled task \(T\), let \(\mathcal G_D\) be the set of global assignments that
can be generated by the chosen local machine model within depth \(D\).  Define
the depth-restricted covering number
\[
  \chi_G^D(T)
  =
  \min\Bigl\{
    K:\exists g_1,\ldots,g_K\in\mathcal G_D
    \text{ covering all contexts of }T
  \Bigr\}.
\]
Then any cover-admissible classical simulator with communication \(B\), memory
\(M\), and local depth \(D\) must satisfy
\[
  2^{B+M}\geq \chi_G^D(T),
\]
or equivalently
\[
  B+M\geq \log_2\chi_G^D(T).
\]
If one wants \(D\) explicitly as a lower-bound variable, define the inverse
trade-off function
\[
  D_T(K)=\min\{D:\chi_G^D(T)\leq K\}.
\]
Then every such simulator also satisfies
\[
  D\geq D_T(2^{B+M}).
\]
In this form, the Turing-machine reading is no longer only motivational: it
specifies how local computation depth trades against communication and memory
inside the same covering-number obstruction.

\subsection*{Bridge to chart-restricted global contextuality}

The following bridge concerns the restricted chart model; the general
classical statement is Theorem~\ref{thm:general-separator}.  In a measurement
scenario, let \(V\) be the set of all measurement labels and let \(\C\) be the
family of compatible contexts.  A deterministic global noncontextual chart is
a function
\[
  g\in \prod_{v\in V} A_v,
\]
assigning one outcome to every measurement, independently of the context in
which that measurement later appears.  A \(\GNCHV\) model is a random choice of
such charts before the measurement context is used.  Thus, after all boundary
transcripts, memory states, and bounded-depth local computations have been
fixed, a classical noncontextual simulator contributes exactly one chart
\[
  g_{t,s}\in\mathcal G_D.
\]
The pair \((t,s)\) is the classical coordination data selecting this chart.

Local noncontextuality and \(\GLHV\) say that the small pieces are classically
explainable: each local subsystem, and even each observed multipartite block,
admits an appropriate classical description.  Global contextuality says that
these descriptions cannot be chosen as restrictions of one common global
chart.  In other words, the obstruction is not local computation inside one
block, and not Bell nonlocal communication between the tested parties; it is a
failure of global gluing.

For a support or relation version of the empirical model, define the support
task \(T_p\) by
\[
  W_C=\operatorname{supp}(p_C)
  \qquad (C\in\C).
\]
If one chart \(g\) satisfies \(g|_C\in W_C\) for every context, then the
support data admit a zero-cost global explanation.  If no such chart exists,
then a classical simulator can only succeed by selecting different charts for
different contexts.  The number of charts available to a simulator with
communication \(B\), memory \(M\), and local depth \(D\) is at most
\[
  2^{B+M}
\]
inside \(\mathcal G_D\).  Hence support-level global contextuality gives the
resource inequality
\[
  2^{B+M}\geq \chi_G^D(T_p).
\]
For probabilistic, rather than support-level, global contextuality, the same
bridge is replaced by an approximate version: distance to the \(\GNCHV\) polytope,
contextual fraction, or fractional covering measures how much
context-dependent chart selection is still required.

\section{Empirical Models, Tasks, and Classical Charts}

We use a finite relation-task presentation.  Let \(\C\) be a finite set of
contexts.  Each context \(C\in\C\) has a set of possible outcomes \(O_C\) and a
winning set
\[
  W_C \subseteq O_C.
\]
A strategy for the task produces, for each context \(C\), a probability
distribution \(p_C\) on \(O_C\).  It wins context \(C\) with probability
\[
  p_C(W_C)=\sum_{a\in W_C} p_C(a).
\]

When the contexts arise from compatible measurements, the family
\[
  p=\{p_C\}_{C\in\C}
\]
is an empirical model.  A perfect quantum strategy satisfies
\[
  p_C(W_C)=1
  \qquad\text{for all } C\in\C.
\]

\begin{definition}[Global assignment]
A global assignment is a function \(g\) assigning an outcome to every basic
measurement appearing in the scenario.  Its restriction to a context \(C\) is
denoted \(g|_C\in O_C\).  We say that \(g\) covers \(C\) if
\[
  g|_C\in W_C.
\]
\end{definition}

If one global assignment covers all contexts, the relation task has a
deterministic global noncontextual solution.  Contextuality appears when no
single global assignment covers all contexts.

For probabilistic empirical models, a global noncontextual hidden-variable
model, abbreviated \(\GNCHV\), is a probability distribution \(\lambda\) over
global assignments such that
\[
  p_C(a)
  =
  \sum_{g:\,g|_C=a}\lambda(g)
  \qquad
  (C\in\C,\;a\in O_C).
\]
Thus a \(\GNCHV\) model is a context-independent random choice of one global
chart.  A local subsystem is noncontextual when its restricted empirical model
admits the analogous representation on that subsystem.

In a multipartite scenario, let \(C_i\) be the compatible local context tested
at party \(i\), let \(\mathbf a_{C_i}\) be its joint block outcome, and write
\(C=(C_1,\ldots,C_r)\).  We use \(\GLHV\) for a generalized Bell-local
hidden-variable model of the form
\[
  p_C(\mathbf a_{C_1},\ldots,\mathbf a_{C_r})
  =
  \sum_\lambda q_\lambda
  \prod_{i=1}^r
  p_i(\mathbf a_{C_i}\mid C_i,\lambda),
\]
where the same distribution \(q_\lambda\) is used for all tested multipartite
contexts.  Crucially, the joint block response
\(p_i(\mathbf a_{C_i}\mid C_i,\lambda)\) may depend on \(C_i\) and need not
admit a context-independent value assignment to the individual measurements
inside the block.  Requiring such a refinement would be an additional
conditional local-\(\NCHV\) assumption, not part of \(\GLHV\) itself.
Thus local \(\NCHV\), blockwise \(\GLHV\), and global \(\GNCHV\) are three
different conditions: local \(\NCHV\) concerns unconditional subsystem
marginals, \(\GLHV\) concerns factorization across parties while allowing
arbitrary within-party compatible-block responses, and \(\GNCHV\) requires
one context-independent assignment over all basic measurements.  In the
worked examples below, the local-\(\NCHV\) and \(\GLHV\) certificates are
imported from Ref.~\cite{yang2026global}.

\section{From Games to General Tasks}

The previous section is phrased in the language of relation tasks, which makes
it resemble Bell games or KS parity games.  The purpose of this section is to
explain how the same language applies to a broad class of ordinary tasks.  The
key move is to view a task as a local constraint system.

\begin{definition}[Constraint compilation of a task]
Let \(T\) be a finite task whose candidate solutions are assignments
\[
  y\in \prod_{v\in V} A_v
\]
to variables \(V\).  A constraint compilation of \(T\) is a hypergraph
\[
  H_T=(V,E)
\]
together with, for each hyperedge \(e\in E\), an allowed local relation
\[
  R_e \subseteq \prod_{v\in e} A_v.
\]
A global assignment \(y\) solves the compiled task when
\[
  y|_e \in R_e
  \qquad\text{for all } e\in E.
\]
\end{definition}

This is the usual constraint-satisfaction form.  It includes Boolean CSPs,
graph coloring, graph homomorphism, database consistency constraints, local
checks in distributed protocols, and many verification problems.  In this
representation, each hyperedge \(e\) is a context:
\[
  C_e=e,\qquad W_{C_e}=R_e.
\]
Thus the task \(T\) induces a relation task
\[
  \mathsf{Rel}(T)=\{(C_e,W_{C_e})\}_{e\in E}.
\]

\begin{definition}[Task-induced empirical model]
A physical or algorithmic procedure for \(T\) induces an empirical model
\[
  p_T=\{p_e\}_{e\in E},
\]
where \(p_e\) is the distribution of local outputs on the variables in
hyperedge \(e\).  The procedure locally satisfies the compiled task with
probability one when
\[
  p_e(R_e)=1
  \qquad\text{for every } e\in E.
\]
\end{definition}

This turns the question ``does the task have a globally consistent classical
solution?'' into the question ``does the locally observed empirical model admit
a global classical explanation?''  The distinction is important:

\begin{itemize}[leftmargin=2em]
\item If there is one global assignment \(y\) satisfying all \(R_e\), then the
compiled task is classically satisfiable.
\item If local empirical distributions satisfy every \(R_e\) but cannot be
glued into a single global hidden-variable model, then the model exhibits a
contextual obstruction.
\item If the obstruction only appears after combining otherwise classical
local and blockwise explanations, it is a candidate for global contextuality.
\end{itemize}

The framework is therefore not restricted to hand-built games.  It applies to
any task that admits a useful local-constraint compilation.  The price of this
generality is that it does not apply uniformly to all function-computation
tasks.  A task of the form
\[
  f(x_1,\ldots,x_k)
\]
where each party must output the same global value cannot be solved without
communicating the relevant input information; entanglement or contextuality
does not violate no-signalling.  The natural quantum advantage tasks in this
framework are instead relation, consistency, and sampling tasks, where local
outputs are not required to reveal all remote inputs but must jointly satisfy a
global relation.

\subsection{General coordination problem}

For a compiled task \(T\), define the classical coordination region by applying
the cost definition to the induced empirical model:
\[
  \Cost_C^\epsilon(T)
  :=
  \Cost_C^\epsilon(p_T).
\]
The central question becomes:
\[
  \text{How far from the origin is } \Cost_C^\epsilon(T)?
\]
Equivalently, how much communication, memory, or local computation is required
for a classical system to coordinate local views into a globally consistent
behavior?

The covering numbers below now have a direct task-theoretic meaning.  The
quantity \(\chi_G(T)\) counts how many global classical explanations are needed
to cover all local constraints of the compiled task.  If one explanation does
not suffice, the classical protocol must coordinate among several explanations
using communication, memory, computation, or a mixture of these resources.

\begin{proposition}[Compiled-task covering bound]
Let \(T\) be a compiled task with contexts \(e\in E\), relations \(R_e\), and
context distribution \(\mu\).  Suppose a classical protocol for \(T\) has at
most \(K\) possible coordination states.  Assume that, after fixing its
randomness and one coordination state \(s\), the protocol induces a
deterministic global assignment \(g_s\).  If the protocol succeeds with
probability at least \(\eta\), then the assignments
\[
  \{g_s: s=1,\ldots,K\}
\]
cover an \(\eta\)-fraction of the compiled constraints:
\[
  \eta
  \leq
  \Pr_{e\sim\mu}\big[\exists s\; g_s|_e\in R_e\big].
\]
Consequently, if every single global assignment satisfies at most a
\(\beta_G(T)\)-fraction of the constraints, then
\[
  K \geq \frac{\eta}{\beta_G(T)}.
\]
For \(m\) independent copies, the same argument gives
\[
  K \geq \frac{\eta_m}{\beta_G(T)^m},
\]
where \(\eta_m\) is the target success probability on the repeated task.
\end{proposition}

\begin{proof}
Fix the protocol randomness.  If the average success over the randomness is at
least \(\eta\), then some fixed choice of randomness also has success at least
\(\eta\).  For that fixed choice, each coordination state \(s\) determines one
global assignment \(g_s\).  A context \(e\) can be won only if at least one of
these assignments satisfies \(R_e\); otherwise no possible coordination state
can output a locally valid answer on \(e\).  This proves the covering
inequality.  Since a single assignment covers at most \(\beta_G(T)\) of the
contexts, \(K\) assignments cover at most \(K\beta_G(T)\) by the union bound,
so \(K\beta_G(T)\geq\eta\).  On \(m\) independent copies, any deterministic
assignment restricts to one assignment on each copy.  Under the product
context distribution its success is therefore at most \(\beta_G(T)^m\),
giving \(K\beta_G(T)^m\geq\eta_m\).
\end{proof}

\subsection{Advantage scale: states versus bits}

This proposition separates two notions of quantum advantage.  Let
\[
  K_C^\eta(T)
\]
be the minimum number of classical coordination states needed to reach success
probability at least \(\eta\), and let
\[
  \kappa_C^\eta(T)=\log_2 K_C^\eta(T)
\]
be the same cost measured in bits.  If a quantum procedure wins one copy with
probability \(\omega_Q(T)\), while one global classical assignment wins at most
\(\beta_G(T)<\omega_Q(T)\), then matching the \(m\)-copy quantum success
\(\omega_Q(T)^m\) requires
\[
  K_C^{\omega_Q^m}(T^m)
  \geq
  \left(\frac{\omega_Q(T)}{\beta_G(T)}\right)^m
\]
and therefore
\[
  \kappa_C^{\omega_Q^m}(T^m)
  \geq
  m\log_2\!\left(\frac{\omega_Q(T)}{\beta_G(T)}\right).
\]

Thus the repeated-task bound is exponential in the raw number of classical
coordination states, but linear in the number of communication or memory bits.
For the seven-test Hardy task studied below,
\[
  \frac{\omega_Q}{\beta_G}=1+\frac{q}{6},
\]
so the lower bound is
\[
  K_C\geq (1+q/6)^m,
  \qquad
  B+M\geq m\log_2(1+q/6).
\]
At \(q=1/4\), this is
\[
  K_C\geq (25/24)^m,
  \qquad
  B+M\geq 0.0589\,m.
\]

This is already a genuine support-chart coordination advantage, but it is not
an exponential lower bound in bits.  To obtain a quadratic or exponential
bit-complexity separation through the covering-number route, one needs a
family of compiled tasks
\(\{T_n\}\) whose global contextual covering number grows faster with the
problem size:
\[
  \log_2 \chi_G(T_n)=\Omega(n^2)
  \quad\text{gives a quadratic bit lower bound,}
\]
whereas
\[
  \log_2 \chi_G(T_n)=\Omega(2^n)
  \quad\text{would give an exponential bit lower bound.}
\]
The memory-cost literature for stabilizer-like quantum subtheories suggests
that quadratic bit growth is realistic at the more general causal-state
level~\cite{karanjai2018contextuality}.  The seven-test Hardy projection does
not prove such a support-cover scaling; it gives a constant chart rate under
parallel repetition.

\subsection{Examples of compiled tasks}

\begin{enumerate}[leftmargin=2em]
\item \textbf{CSP and graph problems.}  Variables are vertices or logical
variables, and hyperedges are clauses, edges, or local graph constraints.
Classical coordination cost measures the resources required to make local
constraint choices globally consistent.

\item \textbf{Distributed consistency.}  Variables are local replicas, log
positions, or transaction states.  Hyperedges encode local consistency,
serializability, or agreement checks.  Communication and memory are the usual
resources used to maintain global consistency.

\item \textbf{MBQC patterns.}  Variables are local measurement outcomes and
contexts are compatible measurement patterns.  Classical feed-forward is a
coordination mechanism, while the resource state supplies nonclassical
correlations.  Contextuality or magic marks the part of the pattern that cannot
be absorbed into a low-cost classical simulation.

\item \textbf{Sampling and generation.}  Variables are local pieces of a
sample, such as tokens, patches, or latent factors.  Hyperedges encode local
semantic, syntactic, or structural constraints.  A contextuality-style analysis
asks whether the local marginals arise from one global latent-variable model or
whether a richer coordination mechanism is required.
\end{enumerate}

\begin{remark}
The last example should not be confused with quantum contextuality unless a
physical measurement scenario is specified.  It is included to mark a possible
generalization: contextuality-like distances to noncontextual latent-variable
models may be useful for generative modeling, but this is a classical or
cognitive analogue unless backed by a genuine quantum empirical model.
\end{remark}

\subsection{Computer-science reading}

The same formalism can be read without quantum terminology.  A relation task
is a local-to-global consistency problem:
\[
  \forall C\in\C,\quad x_C\in W_C
  \qquad\overset{?}{\Longrightarrow}\qquad
  \exists x\in\prod_{v\in V}A_v
  \text{ such that } x|_C=x_C \text{ for all }C.
\]
In the support version used here, a single global chart solves the task iff
\[
  \exists x\in\prod_{v\in V}A_v
  \quad
  \forall C\in\C,\quad
  x|_C\in W_C.
\]
Contextuality is the failure of this implication when all local pieces are
individually meaningful.

For a CSP, \(V\) is the set of variables and each context \(C\) is the scope
of a constraint.  For a database, the local data are relations
\(\{R_C\}_{C\in\C}\), and global consistency asks whether there exists a joint
relation \(J\) such that
\[
  \pi_C(J)=R_C
  \qquad(C\in\C),
\]
where \(\pi_C\) is projection onto the attributes in \(C\).  For a distributed
system, a context is a node neighborhood or local view; the global chart is a
network-wide state whose restrictions match all local views:
\[
  X|_{N_r(v)}=x_{N_r(v)}
  \qquad(v\in V_{\rm net}).
\]
The coordination cost \(\kappa_D(T)=\log_2\chi_G^D(T)\) is therefore the number
of bits needed to select among global configuration templates compatible with
the local views.

For MBQC, the local variables are measurement outcomes and the classical
control data used for feed-forward.  A low-cost classical chart would preassign
outcomes for all measurement branches:
\[
  s_v=s_v(m_v),
  \qquad
  m_v=f_v(x,s_{<v}),
\]
so contextuality marks the part of the measurement pattern that cannot be
absorbed into such a branch-independent classical explanation.

\section{Restricted Chart Covering and Coordination Costs}

We now define the cost region and then study a restricted chart projection of
it.  The cost region itself permits general classical protocols.  The covering
numbers in this section apply only when a transcript-memory state selects a
context-independent global chart; they do not replace
Theorem~\ref{thm:general-separator} for arbitrary causal simulation.  A
classical distributed protocol may use
communication, internal memory, and local computation to generate outputs.
For a protocol \(\Pi\), write
\[
  \mathrm{cost}(\Pi)=(B(\Pi),M(\Pi),D(\Pi)),
\]
where \(B\) is total communication in bits, \(M\) is the maximum amount of
classical internal memory in bits, and \(D\) is a local computation depth or
step bound.  The precise choice of \(D\) depends on the computational model;
our covering theorem below uses only \(B\) and \(M\), and then refines the
covering number by imposing a \(D\)-computability restriction.

The convention is that memory is measured in bits.  If a simulation model is
specified instead by a finite set of \(S_M\) internal memory states, then in
the present notation
\[
  M=\lceil\log_2 S_M\rceil,
\]
and the covering lower bound may equivalently be read as
\[
  B+\log_2 S_M\geq \log_2\chi_G(T)
\]
up to integer rounding.  This is the conversion used when comparing with
memory-cost results stated in terms of the number of hidden states.

The depth parameter \(D\) is a classical postprocessing or chart-generation
depth, not a quantum circuit depth.  Depending on the application it may mean
the depth of a classical circuit generating a global chart, the number of
local update rounds in a simulator, or a restricted family of admissible
postprocessing maps.  The twisted-hypercube example below uses a concrete
fan-in-two XOR-tree depth.

\begin{definition}[Classical cost region]
For an empirical model \(p\) and error tolerance \(\epsilon\), define
\[
 \begin{aligned}
  \Cost^\epsilon_C(p)
  =
  \{(B,M,D):\;&\text{there exists a classical protocol }\Pi\\
              &\text{with cost at most }(B,M,D)
                \text{ and }d(\Pi,p)\leq\epsilon\}.
 \end{aligned}
\]
Here \(d\) may be taken to be worst-case total variation distance over
contexts.
\end{definition}

The communication-complexity literature studies projections of
\(\Cost^\epsilon_C(p)\) onto the \(B\)-axis.  The memory-cost literature
studies projections onto the \(M\)-axis.  The present proposal is to study the
full region, especially for globally contextual empirical models.

\subsection{Global contextual covering number}

\begin{definition}[Global contextual covering number]
For a relation task \(T=(\C,\{W_C\}_{C\in\C})\), define
\[
  \chi_G(T)=
  \min\Bigl\{
    K:\exists g_1,\ldots,g_K
    \text{ such that for every } C\in\C,
    \exists j \text{ with } g_j|_C\in W_C
  \Bigr\}.
\]
\end{definition}

Thus \(\chi_G(T)=1\) exactly when one global assignment satisfies all contexts.
If \(\chi_G(T)>1\), a classical explanation must switch among several global
assignments.

\begin{definition}[Coordination bits]
For a depth-restricted chart class \(\mathcal G_D\), define
\[
  \kappa_D(T)=\log_2\chi_G^D(T),
  \qquad
  \kappa(T)=\log_2\chi_G(T).
\]
We call \(\kappa_D(T)\) the \(D\)-restricted coordination information of the
task, measured in \emph{coordination bits} or \emph{coord-bits}.  One
coord-bit is one binary distinction needed to select among globally
noncontextual charts.  If an integer number of bits is required, the
operational cost is \(\lceil \kappa_D(T)\rceil\).
\end{definition}

For a cover-admissible simulator, communication and memory contribute directly
to this chart-selection budget:
\[
  B+M\geq \kappa_D(T).
\]
Local computation depth is not itself measured in bits, but it changes the
available chart class.  The number of coord-bits saved by allowing depth \(D\)
instead of depth \(0\) is
\[
  \Delta_D(T)
  =
  \kappa_0(T)-\kappa_D(T)
  =
  \log_2\!\left(\frac{\chi_G^0(T)}{\chi_G^D(T)}\right),
\]
whenever the two covering numbers are finite.

\begin{definition}[Cover-admissible protocol]
A deterministic protocol is cover-admissible for \(T\) if each possible public
transcript \(t\) and coordination-memory state \(s\) determines a
context-independent global assignment
\[
  g_{t,s}\in \prod_{v\in V} A_v,
\]
and, when queried on a context \(C\), the protocol's remaining local output
rule is the restriction \(g_{t,s}|_C\).  In this model, communication and memory
act as a switch among global classical charts.
\end{definition}

\begin{theorem}[Covering lower bound]
Suppose a deterministic cover-admissible protocol solves the relation task
\(T\) perfectly using \(B\) bits of public communication transcript and \(M\)
bits of classical coordination memory.  Then
\[
  B+M \geq \log_2 \chi_G(T).
\]
\end{theorem}

\begin{proof}
The communication transcript and the memory state together distinguish at most
\[
  2^{B+M}
\]
coordination states.  Once such a coordination state is fixed, the protocol's
deterministic output rule induces one global assignment \(g_j\).  If the
protocol solves the task for every context, then the induced assignments must
cover every context.  Hence the number of induced assignments is at least
\(\chi_G(T)\).  Therefore
\[
  2^{B+M}\geq \chi_G(T),
\]
and taking logarithms proves the claim.
\end{proof}

\begin{remark}
The theorem is deliberately elementary.  Its role is to isolate the
combinatorial core of chart selection.  It is not a theorem about all
classical protocols.  Randomized or adaptive protocols that do not reduce to
a context-independent chart after the counted boundary state is fixed must be
analyzed by Theorem~\ref{thm:general-separator} or by the causal realization
bound~\eqref{eq:causal-rank-bound}.

The memory convention matters.  The cover-admissible model treats memory as a
resource that can index one of several global charts after the relevant
coordination information has been fixed.  This is close to a communication
transcript or a reloadable coordination state.  It is not identical to the
sequential memory model used in some KS simulation work, where a hidden memory
state may be prepared before the measurement context is revealed and then
updated during a sequence of measurements.  Relating that automaton-style
memory model to \(\chi_G(T)\) requires a separate operational reduction.

The same caveat applies to adaptive protocols such as MBQC with classical
feed-forward.  Such a protocol is covered by the elementary theorem only after
the transcript includes the relevant branch information and fixing that
transcript plus memory state determines one context-independent chart.  If the
protocol uses genuinely branch-dependent updates that cannot be reduced to a
fixed chart after conditioning on the transcript, it lies outside the
cover-admissible model and requires a different lower-bound argument.
\end{remark}

\subsection{Support-level and probabilistic coordination costs}

The covering number is a support-level invariant.  Given a probabilistic
empirical model \(p=\{p_C\}_{C\in\C}\), let \(T_p\) be the induced support task
defined by
\[
  W_C=\operatorname{supp}(p_C)
  =
  \{o:p_C(o)>0\}.
\]
The quantity \(\chi_G^D(T_p)\) records only which local events are possible.  It
forgets the probabilities assigned to those events.

For probabilistic simulation one needs a stronger quantity.  Fix a class
\(\mathcal P\) of classical simulators and a depth-restricted chart or response
class \(\mathcal G_D\).  Define
\[
  K_{\mathcal P}^D(p)
\]
to be the minimum number of internal classical states, coordination labels, or
charts required by a simulator in \(\mathcal P\), with local computation depth
at most \(D\), to reproduce \(p\) exactly.  The associated memory cost is
\[
  M_{\mathcal P}^D(p)
  =
  \left\lceil \log_2 K_{\mathcal P}^D(p)\right\rceil .
\]
Thus the support-level bound
\[
  B+M\geq \log_2\chi_G^D(T)
\]
has an operational probabilistic analogue
\[
  B+M\geq \log_2 K_{\mathcal P}^D(p),
\]
whenever a \(B\)-bit transcript together with an \(M\)-bit memory state can
select at most \(2^{B+M}\) simulator states in \(\mathcal P\).  Equivalently,
one may define
\[
  D_p(K)=\min\{D:K_{\mathcal P}^D(p)\leq K\},
\]
so that a simulator with only \(K=2^{B+M}\) coordination states must satisfy
\[
  D\geq D_p(2^{B+M}).
\]

\begin{proposition}[Support projection]
For chart-based simulator classes in which each internal state determines one
depth-\(D\) global chart, exact simulation of \(p\) with
\(K_{\mathcal P}^D(p)\) simulator states induces a depth-\(D\) cover of the
support task \(T_p\).  Hence
\[
  \chi_G^D(T_p)\leq K_{\mathcal P}^D(p).
\]
The converse implication need not hold.
\end{proposition}

\begin{proof}
An exact chart-based simulator for \(p\) specifies a finite family of global
charts, together with weights, whose mixture reproduces all probabilities
\(p_C(o)\).  Forgetting the weights and retaining only which local events can
be produced gives a support-level cover of \(T_p\).  Therefore the number of
charts needed to cover the support cannot exceed the number of chart states
needed for exact probabilistic simulation.  Conversely, a small family of
charts may cover all possible events without reproducing their required
probability weights, so the reverse inequality is false in general.
\end{proof}

This distinction matches the hierarchy in the sheaf-theoretic analysis of
contextuality~\cite{abramsky2011sheaf}.  The invariant \(\chi_G^D\) lives at
the possibilistic or support level; it quantifies the failure of local supports
to be explained by a small family of global sections.  The invariant
\(K_{\mathcal P}^D\) lives at the probabilistic or operational simulation
level; it also records the weights, updates, and response structure needed to
match the empirical model.  Thus a small support covering number does not rule
out a large probabilistic simulation cost.  In particular, the stabilizer-memory
lower bounds used below are lower bounds on \(K_{\mathcal P}^D\), or on the
corresponding memory bits, not on \(\chi_G^D\).  The term ``strong separation''
below is used in the complexity-theoretic sense of asymptotic resource growth,
not necessarily in the Abramsky--Brandenburger sense of strong contextuality.

\subsection{Resource-sensitive covers}

The ordinary covering number only sees the total number of available charts.
It therefore gives a sum bound in \(B+M\), not a genuine separation between
communication and memory.  To expose such trade-offs one must keep track of
how charts are selected.

\begin{definition}[Structured chart selector]
A selector model \(\mathfrak S\) assigns to every communication budget \(B\) a
set \(\mathfrak S_B\) of admissible transcript maps
\[
  \sigma:\C\to\{1,\ldots,2^B\}.
\]
The choice of \(\mathfrak S_B\) encodes the operational communication model.
For example, in a multipartite communication problem \(\mathfrak S_B\) should
come from \(B\)-bit protocols, not from arbitrary functions on \(\C\).
\end{definition}

\begin{definition}[Resource-sensitive covering region]
Fix a selector model \(\mathfrak S\) and a depth-restricted chart class
\(\mathcal G_D\).  A relation task \(T=(\C,\{W_C\})\) has a
\((B,M,D)\)-structured cover if there exist memory labels
\[
  s\in\{1,\ldots,2^M\},
\]
selectors \(\sigma_s\in\mathfrak S_B\), and charts
\[
  g_{s,t}\in\mathcal G_D
  \qquad
  (s=1,\ldots,2^M,\; t=1,\ldots,2^B)
\]
such that for every context \(C\in\C\),
\[
  \exists s
  \quad
  g_{s,\sigma_s(C)}|_C\in W_C.
\]
Define the structured coordination region
\[
  \mathsf R_{\mathfrak S}(T)
  =
  \{(B,M,D): T \text{ has a }(B,M,D)\text{-structured cover}\}.
\]
\end{definition}

This definition is deliberately model-dependent: different choices of
\(\mathfrak S_B\) give different notions of communication.  That dependence is
a feature rather than a bug, because a true \(B\)-versus-\(M\) theorem cannot
be obtained from an unstructured count alone.

\begin{proposition}[Collapse under unrestricted selectors]
If \(\mathfrak S_B\) contains every function
\(\C\to\{1,\ldots,2^B\}\), then
\[
  (B,M,D)\in\mathsf R_{\mathfrak S}(T)
  \quad\Longleftrightarrow\quad
  \chi_G^D(T)\leq 2^{B+M}.
\]
\end{proposition}

\begin{proof}
If a structured cover exists, it uses at most \(2^{B+M}\) charts in
\(\mathcal G_D\), so these charts form an ordinary depth-restricted cover.
Conversely, suppose \(g_1,\ldots,g_K\) cover \(T\), with
\(K\leq2^{B+M}\).  Index the \(K\) charts by pairs \((s,t)\).  Because
selectors are unrestricted, for each memory label \(s\) we may choose a map
\(\sigma_s\) that sends every context assigned to a chart in row \(s\) to the
corresponding transcript label.  Contexts not assigned to row \(s\) may be
sent arbitrarily.  This realizes the ordinary cover as a structured cover.
\end{proof}

Thus the present lower bound is exactly the unstructured projection of a more
refined question.  To prove that communication and memory are not
interchangeable, one must impose an operationally meaningful selector class
\(\mathfrak S_B\), such as rectangle selectors in a two-party communication
task or protocol-tree selectors in a distributed task.

\subsection{Two-party rectangle selectors}

The standard two-party case makes the previous paragraph concrete.  Suppose
the contexts are indexed by pairs
\[
  (x,y)\in X\times Y,
\]
where Alice knows \(x\) and Bob knows \(y\).  A deterministic \(B\)-bit
communication protocol induces a transcript map
\[
  \sigma:X\times Y\to\mathcal T,
  \qquad |\mathcal T|\leq 2^B.
\]
For every transcript \(t\), the fiber
\[
  \sigma^{-1}(t)
\]
is a combinatorial rectangle \(A_t\times B_t\subseteq X\times Y\).  This is
the usual rectangle property of deterministic communication
protocols~\cite{kushilevitz1997communication}.

\begin{definition}[Global-chart rectangle cover]
For a two-party relation task \(T\) and a depth-restricted chart class
\(\mathcal G_D\), define \(\operatorname{rect}_G^D(T)\) to be the minimum
number \(L\) such that \(X\times Y\) can be covered by rectangles
\[
  R_1,\ldots,R_L\subseteq X\times Y
\]
and charts
\[
  g_1,\ldots,g_L\in\mathcal G_D
\]
with the property that for every \((x,y)\in R_\ell\),
\[
  g_\ell|_{C_{x,y}}\in W_{x,y}.
\]
\end{definition}

This number refines the ordinary global covering number.  Always
\[
  \chi_G^D(T)\leq \operatorname{rect}_G^D(T),
\]
because every rectangle cover is in particular a cover by charts.  The
inequality can be strict when the set of contexts won by one chart is
nonrectangular and must be decomposed into several communication rectangles.

\begin{proposition}[Rectangle lower bound]
For a two-party relation task \(T\), any deterministic cover-admissible
protocol with no coordination memory and local chart depth \(D\) that solves
\(T\) perfectly using \(B\) bits of communication satisfies
\[
  B\geq \log_2 \operatorname{rect}_G^D(T).
\]
\end{proposition}

\begin{proof}
Each transcript leaf of the protocol is a rectangle in \(X\times Y\).  Since
the protocol is cover-admissible, fixing that transcript also fixes a
depth-\(D\) global chart.  On every context in the transcript rectangle, that
chart must satisfy the corresponding relation.  Thus the protocol leaves form
a global-chart rectangle cover with at most \(2^B\) rectangles.  Therefore
\(2^B\geq \operatorname{rect}_G^D(T)\).
\end{proof}

With memory, the analogous object is a collection of at most \(2^M\) chart
libraries, each addressed by a \(B\)-bit rectangle selector.  This gives a
genuine place to formulate \(B\)-versus-\(M\) questions: communication
restricts the geometry of the selector fibers, while memory controls how many
such chart libraries are available.  The present paper does not compute a
nontrivial separation for this refined region, but the definition identifies
the next combinatorial object that must be analyzed.

\subsection{Depth-restricted covering: a first calculation}

The framework above defines inverse-depth quantities such as \(D_T(K)\), but
a definition is not yet a calculation.  We now give a small support-level gluing task where the
\((B,M,D)\) trade-off can be computed exactly.  This example is not meant to be
a physical global-KS construction; its role is to show how the computation
axis enters the same covering machinery.

Let \(n\geq2\).  For each vertex \(x\in\{0,1\}^n\) of the \(n\)-dimensional
hypercube, introduce one binary variable \(v_x\).  For each edge in direction
\(i\), written
\[
  C_{x,i}=\{v_x,v_{x\oplus e_i}\}
  \qquad (x_i=0),
\]
impose the parity relation
\[
  a\oplus b=1,
\]
except on the distinguished edge
\[
  C_{0,1}=\{v_{0^n},v_{e_1}\},
\]
where the relation is instead
\[
  a\oplus b=0.
\]
Call this compiled task \(Q_n\).

\begin{proposition}[Twisted hypercube gluing obstruction]
The task \(Q_n\) has no single global satisfying assignment.
\end{proposition}

\begin{proof}
Consider the square with vertices
\[
  0^n,\quad e_1,\quad e_2,\quad e_1\oplus e_2.
\]
The distinguished edge requires parity \(0\), while the other three edges of
the square require parity \(1\).  XORing the four edge equations around the
cycle, every vertex value appears twice on the left and cancels, giving \(0\).
The right side is \(0\oplus1\oplus1\oplus1=1\), a contradiction.
\end{proof}

Now restrict the allowed charts.  Let \(\mathcal G_D\) consist of affine parity
charts
\[
  g_{S,c}(x)=c\oplus\bigoplus_{i\in S}x_i,
  \qquad c\in\{0,1\},
\]
with
\[
  |S|\leq 2^D.
\]
This is the natural fan-in-two XOR-tree model: depth \(D\) can combine at most
\(2^D\) input coordinates into one parity value.

\begin{proposition}[Depth-restricted covering number]
For the twisted hypercube task \(Q_n\),
\[
  \chi_G^D(Q_n)
  =
  \max\left\{2,\left\lceil\frac{n}{2^D}\right\rceil\right\}.
\]
Consequently, for \(K\geq2\),
\[
  D_{Q_n}(K)
  =
  \max\left\{0,\left\lceil\log_2\frac{n}{K}\right\rceil\right\},
\]
while \(D_{Q_n}(1)=\infty\).
\end{proposition}

\begin{proof}
For an affine parity chart \(g_{S,c}\),
\[
  g_{S,c}(x)\oplus g_{S,c}(x\oplus e_i)
  =
  \begin{cases}
  1,& i\in S,\\
  0,& i\notin S.
  \end{cases}
\]
Thus \(g_{S,c}\) covers every ordinary edge in direction \(i\) exactly when
\(i\in S\), and it covers the twisted edge \(C_{0,1}\) exactly when
\(1\notin S\).

If \(K\) charts cover \(Q_n\), their sets \(S_1,\ldots,S_K\) must cover all
\(n\) coordinate directions, because ordinary edges occur in every direction.
Since each \(|S_j|\leq2^D\), we need
\[
  K2^D\geq n.
\]
Moreover, some chart must contain direction \(1\) to cover the ordinary
direction-\(1\) edges, while some chart must omit direction \(1\) to cover the
twisted edge.  Hence \(K\geq2\).  This gives the lower bound
\[
  K\geq
  \max\left\{2,\left\lceil\frac{n}{2^D}\right\rceil\right\}.
\]

For the matching upper bound, choose
\(K=\max\{2,\lceil n/2^D\rceil\}\).  If \(\lceil n/2^D\rceil\geq2\), partition
the directions \(\{1,\ldots,n\}\) into \(K\) subsets of size at most \(2^D\),
with direction \(1\) in one subset.  Another subset omits direction \(1\), so
it covers the twisted edge.  The corresponding parity charts cover all
ordinary directions and the twisted edge.  If \(\lceil n/2^D\rceil=1\), use
one chart with \(S=\{1,\ldots,n\}\) to cover all ordinary edges and one chart
with \(S=\emptyset\) to cover the twisted edge.  This proves the formula for
\(\chi_G^D(Q_n)\), and the expression for \(D_{Q_n}(K)\) follows by inverting
the inequality \(\chi_G^D(Q_n)\leq K\).
\end{proof}

For a cover-admissible simulator, \(K\leq2^{B+M}\).  Therefore, if
\(B+M\geq1\), the twisted hypercube gives the explicit trade-off
\[
  D
  \geq
  \max\left\{
    0,\left\lceil \log_2 n-(B+M)\right\rceil
  \right\}.
\]
With no coordination bit, \(K=1\), perfect support-level simulation is
impossible for any depth because the task has no global section.  Each extra
coordination bit doubles the number of available charts and can reduce the
required local parity depth by at most one.

The twisted hypercube is only a combinatorial model of gluing failure.  Its
purpose is to make the \(D\)-axis computable in a transparent setting.  In a
genuine global-KS example, such as the Hardy construction studied later, the
task \(T_p\) is obtained from the supports of a physical empirical model and
the admissible charts are \(\GNCHV\) assignments.  The same depth-restricted
covering calculation is then the object one would need to compute.

\subsection{A cautionary test: Tseitin contradictions}

It is tempting to look for stronger examples among classical hard
contradictions.  Tseitin systems on expander graphs are a natural first
candidate: they are mod-2 parity contradictions and have strong proof
complexity lower bounds~\cite{urquhart1987hard,bensasson2001short}.  However,
they do not make the global covering number large.  This illustrates an
important difference between proof complexity and coordination covering.

Let \(G=(V,E)\) be a connected graph with \(|V|\geq2\).  A charge function is a
map
\[
  c:V\to\{0,1\}.
\]
Assume it has odd total charge,
\[
  \bigoplus_{v\in V} c(v)=1.
\]
The Tseitin task \(T_{\rm Ts}(G,c)\) has one binary variable \(x_e\) for each
edge \(e\in E\).  For each vertex \(v\), the context is the incident edge set
\[
  C_v=\{e\in E:e\ni v\},
\]
and the winning relation is
\[
  \bigoplus_{e\ni v} x_e=c(v).
\]

\begin{proposition}[Tseitin contradictions have small cover number]
For every connected \(G\) with odd total charge,
\[
  \chi_G(T_{\rm Ts}(G,c))=2.
\]
Moreover, under the uniform distribution on vertices,
\[
  \beta_G(T_{\rm Ts}(G,c))=\frac{|V|-1}{|V|}.
\]
\end{proposition}

\begin{proof}
No global assignment satisfies all vertex equations: XORing all equations
makes every edge variable appear twice on the left, so the left side is \(0\),
whereas the right side is the odd total charge \(1\).  Hence
\(\chi_G>1\).

Fix any vertex \(v_0\).  Remove the equation at \(v_0\).  Since \(G\) is
connected, the remaining incidence system over \(\mathbb F_2\) has full row
rank \(|V|-1\), and hence has a solution for every right-hand side.  Any such
solution satisfies all equations except the omitted one, which must then be
violated by the odd-total-charge argument.  Thus for every \(v_0\) there is a
global assignment whose unique failed context is \(C_{v_0}\).

Choose two distinct vertices \(v_1,v_2\).  The two assignments whose unique
failed contexts are \(C_{v_1}\) and \(C_{v_2}\) cover every vertex context, so
\(\chi_G\leq2\).  Hence \(\chi_G=2\).  The same single-defect assignments show
that \(\beta_G\geq(|V|-1)/|V|\), while inconsistency implies that every
assignment fails at least one vertex equation, giving the reverse inequality.
\end{proof}

Thus large proof-complexity lower bounds for Tseitin formulas do not imply
large global contextual covering numbers.  Proof complexity asks how hard it
is to derive a contradiction; \(\chi_G\) asks how many global charts are needed
to cover the local constraints.  A contradiction can be proof-theoretically
hard while being cover-theoretically easy.

The right preliminary screen for large \(\chi_G\) is a fractional covering
quantity.  Let \(\Delta(\mathcal G)\) denote distributions over deterministic
global assignments, and define
\[
  \gamma_G(T)
  =
  \max_{\nu\in\Delta(\mathcal G)}
  \min_{C\in\C}
  \Pr_{g\sim\nu}[g|_C\in W_C].
\]
This is the largest coverage probability that can be guaranteed uniformly
over contexts by randomizing over global charts.
It should not be confused with \(\beta_G(T)\): \(\beta_G(T)\) is the best
average success of one chart under a chosen context distribution, whereas
\(\gamma_G(T)\) asks for a distribution over charts that covers every context
with uniformly high probability.

It is also different from the contextual fraction.  For a probabilistic model
\(p\), the noncontextual fraction can be written as
\[
  \operatorname{NCF}(p)
  =
  \max\{\lambda:\; p=\lambda p_{\NC}+(1-\lambda)p'
  \text{ for some }p_{\NC}\in\NC\},
\]
and the contextual fraction is \(1-\operatorname{NCF}(p)\).  This is a
probabilistic distance-to-\(\NC\) quantity.  By contrast, \(\chi_G\) and
\(\gamma_G\) are support-level or relation-level covering quantities.  They
are useful for the elementary bit-counting bounds proved here, while
contextual-fraction methods would be the natural tool for sharper approximate
simulation bounds.

\begin{proposition}[Fractional coverage upper bound]
If \(N=|\C|\) and \(\gamma_G(T)>0\), then
\[
  \chi_G(T)
  \leq
  \left\lceil
    \frac{\ln N+1}{\gamma_G(T)}
  \right\rceil .
\]
\end{proposition}

\begin{proof}
Choose \(K\) independent charts from a distribution \(\nu\) witnessing
\(\gamma_G(T)\) up to an arbitrarily small slack.  For any fixed context
\(C\), the probability that none of the \(K\) charts covers \(C\) is at most
\[
  (1-\gamma_G(T))^K\leq e^{-\gamma_G(T)K}.
\]
By the union bound, the probability that some context is uncovered is at most
\[
  N e^{-\gamma_G(T)K}.
\]
For \(K>(\ln N)/\gamma_G(T)\), this probability is less than one, so there
exists a choice of \(K\) charts covering all contexts.  The displayed bound
absorbs the strict inequality and integer rounding.
\end{proof}

Consequently, any task family with very large \(\chi_G(T_n)\) must have very
small uniform fractional coverage \(\gamma_G(T_n)\), not merely a hard
refutation in a proof system.  This suggests that good candidates should have
delocalized violation structure: every chart should miss many contexts, and
no distribution over charts should cover all contexts with substantial
probability.  Expander-code, locally testable code, or quantum-LDPC-style
constraint systems are natural places to look, but they require separate
analysis.

\section{Repetition, Approximation, and Scaling}

Define the best single-assignment success fraction
\[
  \beta_G(T)=
  \max_g \Pr_{C\sim \mu}[g|_C\in W_C],
\]
where \(\mu\) is a distribution over contexts.  For the uniform distribution,
\(\beta_G(T)\) is the largest fraction of contexts covered by one global
assignment.

For \(m\) independent copies of the task, a deterministic global assignment is
equivalently a tuple of copywise assignments \(g_1,\ldots,g_m\).  The
assignment \(g_i\) wins at most a \(\beta_G(T)\)-fraction of contexts in copy
\(i\), and the product context distribution gives total success at most
\(\beta_G(T)^m\).  If a classical protocol with \(K\) coordination states
succeeds with probability at least
\(1-\epsilon\), then
\[
  K\beta_G(T)^m \geq 1-\epsilon.
\]
Using \(K\leq 2^{B+M}\), we obtain:

\begin{proposition}[Repeated-task lower bound]
For \(m\) independent copies of \(T\), any cover-based classical protocol that
succeeds with probability at least \(1-\epsilon\) satisfies
\[
  B+M \geq m\log_2(1/\beta_G(T))+\log_2(1-\epsilon).
\]
\end{proposition}

For constant \(\epsilon<1\), the leading term is linear in \(m\) whenever
\(\beta_G(T)<1\).

\subsection{Scaling laws for task families}

We now state the general asymptotic reading of the preceding bounds.  Let
\(\{T_n\}\) be a family of relation tasks, where \(n\) is the problem size, and
let \(D_n\) be the allowed local computation depth.  The relevant quantities
are:
\[
  \chi_G^{D_n}(T_n),
  \qquad
  \beta_G(T_n),
  \qquad
  \gamma_G(T_n),
  \qquad
  \operatorname{rect}_G^{D_n}(T_n)
\]
when a two-party communication structure is present.

\begin{proposition}[Perfect-support scaling]
Any cover-admissible classical simulator that realizes the support task
\(T_n\) perfectly with communication \(B_n\), coordination memory \(M_n\), and
local depth \(D_n\) satisfies
\[
  B_n+M_n
  \geq
  \log_2\chi_G^{D_n}(T_n).
\]
Equivalently, for a coordination budget \(K_n=2^{B_n+M_n}\),
\[
  D_n\geq D_{T_n}(K_n).
\]
Thus:
\[
\begin{array}{rcl}
\chi_G^{D_n}(T_n)\geq 2^{\Omega(n)}
&\Longrightarrow&
B_n+M_n=\Omega(n),\\[1mm]
\chi_G^{D_n}(T_n)\geq 2^{\Omega(n^2)}
&\Longrightarrow&
B_n+M_n=\Omega(n^2),\\[1mm]
\chi_G^{D_n}(T_n)\geq 2^{2^{\Omega(n)}}
&\Longrightarrow&
B_n+M_n=2^{\Omega(n)}.
\end{array}
\]
\end{proposition}

\begin{proof}
The first inequality is the covering lower bound applied to the
depth-restricted chart class \(\mathcal G_{D_n}\).  The inverse-depth statement
is exactly the definition of \(D_{T_n}(K_n)\).  The three displayed asymptotic
implications follow by taking base-two logarithms.
\end{proof}

The last line is the important bookkeeping point: an exponential number of
charts gives only a linear bit lower bound.  Exponential bit lower bounds
require a doubly exponential number of charts, or a different lower-bound
mechanism.

\begin{proposition}[Distributional and repeated scaling]
Suppose one global chart wins \(T_n\) with probability at most
\(\beta_n=\beta_G(T_n)\) under a context distribution \(\mu_n\), while the
target physical strategy wins with probability \(\omega_n>\beta_n\).  To match
the \(m\)-copy success probability \(\omega_n^m\), any cover-based classical
simulator satisfies
\[
  B+M
  \geq
  m\log_2\!\left(\frac{\omega_n}{\beta_n}\right).
\]
In particular, if \(\omega_n=1\), then
\[
  B+M\geq m\log_2(1/\beta_n).
\]
\end{proposition}

\begin{proof}
A single chart wins at most a \(\beta_n\)-fraction of one-copy contexts, so a
product chart wins at most \(\beta_n^m\) of \(m\)-copy contexts.  If a
classical simulator has \(K\) coordination states, the union bound gives
success at most \(K\beta_n^m\).  Matching the target success \(\omega_n^m\)
therefore requires \(K\beta_n^m\geq\omega_n^m\).  Since
\(K\leq2^{B+M}\), the result follows.
\end{proof}

Writing \(\beta_n=1-\delta_n\), the leading rate is
\[
  \log_2(1/\beta_n)
  =
  \frac{\delta_n}{\ln 2}+O(\delta_n^2)
  \qquad(\delta_n\to0).
\]
Hence a constant soundness gap gives a linear-in-\(m\) bit lower bound, a
gap \(\delta_n=1/\operatorname{poly}(n)\) gives only
\(m/\operatorname{poly}(n)\), and a one-copy lower bound of order
\(\Omega(n^r)\) requires roughly
\[
  \beta_n\leq 2^{-\Omega(n^r)}
\]
or an equivalent amplification mechanism.

\begin{proposition}[Fractional-coverage obstruction to large \(\chi_G\)]
Let \(N_n=|\C_n|\).  If
\[
  \gamma_G(T_n)\geq \alpha_n>0,
\]
then
\[
  \log_2\chi_G(T_n)
  \leq
  \log_2(\ln N_n+1)+\log_2(1/\alpha_n)+1.
\]
Consequently, any family with
\[
  \log_2\chi_G(T_n)=\Omega(f(n))
\]
must have
\[
  \gamma_G(T_n)
  \leq
  (\ln N_n+1)\,2^{-\Omega(f(n))}.
\]
\end{proposition}

\begin{proof}
This is just the fractional coverage upper bound applied to \(T_n\), followed
by taking logarithms and rearranging.
\end{proof}

This proposition explains why simply increasing the size of a contradictory
constraint system is not enough.  If random global charts can cover every
context with noticeable probability, then the integral cover is small up to a
logarithmic factor.  Strong growth of \(\chi_G\) requires delocalized
violation: no distribution over classical charts should cover every context
with substantial probability.

Finally, in a two-party setting one obtains a genuine communication scaling
law from the rectangle refinement:
\[
  B_n
  \geq
  \log_2 \operatorname{rect}_G^{D_n}(T_n)
\]
for memory-free deterministic cover-admissible protocols.  Thus
\[
  \operatorname{rect}_G^{D_n}(T_n)\geq 2^{\Omega(n)}
  \quad\Longrightarrow\quad
  B_n=\Omega(n),
\]
and stronger rectangle growth gives stronger communication lower bounds.  With
memory, the corresponding question is the structured region
\(\mathsf R_{\mathfrak S}(T_n)\): memory supplies multiple chart libraries,
while communication restricts which chart in a library can be selected by
rectangle-shaped transcript fibers.

These scaling laws give a precise search criterion for strong separations
within the restricted chart hierarchy.  One must find natural
global-contextual task families for which at least one of the quantities
\[
  \chi_G^{D_n}(T_n),\qquad
  \operatorname{rect}_G^{D_n}(T_n),\qquad
  1/\beta_G(T_n),\qquad
  1/\gamma_G(T_n)
\]
grows rapidly with \(n\).  The Hardy and KCBS support examples below prove
positivity in this hierarchy; a chart-level strong separation would require a
family with much faster asymptotic growth.  Theorem~\ref{thm:exact-quadratic}
instead obtains a strong separation through causal state complexity.

The size of the constants is also informative.  Ordinary pseudo-telepathy and
KS examples give linear bit lower bounds under repetition, but often with
larger rates than the genuinely global examples currently known.  The GHZ game
has rate
\[
  \log_2(4/3)\approx 0.415
\]
coord-bits per copy, and the six-context Mermin--Peres support task has rate
\[
  \log_2(6/5)\approx 0.263.
\]
The nonlocal magic-square game similarly gives
\[
  \log_2(9/8)\approx 0.170.
\]
By contrast, the Hardy genuinely global example below gives, at its optimal
parameter,
\[
  \log_2(25/24)\approx 0.0589
\]
coord-bits per copy, while the postselected KCBS example gives
\[
  \log_2(\sqrt5/2)\approx 0.161.
\]
Thus the currently explicit genuinely global support examples are
quantitatively weaker.  The stabilizer lift of
Theorem~\ref{thm:exact-quadratic} is asymptotically stronger, but it is a
causal state-complexity separation rather than a proof of large
support-covering number.

\section{Recovering the Bell and KS Readings}

\subsection{GHZ/Mermin}

In the three-party GHZ game, there are four contexts:
\[
  000,\;011,\;101,\;110.
\]
Any deterministic classical assignment satisfies at most three of the four
parity constraints, while the quantum GHZ strategy satisfies all four.  Thus
\[
  \beta_G=3/4.
\]
For \(m\) repetitions,
\[
  B+M \geq m\log_2(4/3)+\log_2(1-\epsilon).
\]
When memory is treated as fixed or free, this is read as a communication lower
bound.  It is the Bell/nonlocal projection of the coordination-cost framework.

\subsection{Mermin--Peres square}

The Mermin--Peres square has six contexts: three rows and three columns.
Because of the parity contradiction, any deterministic noncontextual
assignment satisfies at most five of the six constraints, whereas quantum
observables satisfy all six.  Hence
\[
  \beta_G=5/6.
\]
For \(m\) repetitions,
\[
  B+M \geq m\log_2(6/5)+\log_2(1-\epsilon).
\]
In a single-system simulation where there is no communication resource, this
becomes a memory lower bound.  It is the KS/contextual projection of the same
coordination-cost framework.

\section{Toward Genuine Global Contextuality}
\label{sec:genuine-global}

The preceding examples are sanity checks: GHZ is naturally read through
nonlocality, and Mermin--Peres through local contextuality.  The new target is
a genuinely global contextual model satisfying:
\[
\begin{array}{ll}
\text{(i)} & \text{local subsystems admit noncontextual explanations},\\
\text{(ii)} & \text{observed multipartite blocks admit generalized Bell-local explanations},\\
\text{(iii)} & \text{the whole model admits no global noncontextual hidden-variable model}.
\end{array}
\]
This is a stratified condition on a physical system.  It is stronger than the
bare statement that a measurement cover has no global section: the relevant
local restrictions and the observed multipartite blocks are required to be
classically explainable, and only their union fails to admit one
context-independent global hidden-variable account.  Thus the adjective
``global'' here refers to the whole physical system relative to its subsystem
and block structure, while the covering numbers below still use the standard
global-chart language over the associated measurement scenario.

For such a model, the same quantities \(\chi_G\) and \(\beta_G\) can be
computed relative to global noncontextual assignments.  If \(\beta_G<1\), the
same repeated-task argument gives
\[
  B+M \geq m\log_2(1/\beta_G)+\log_2(1-\epsilon).
\]
If one also restricts local computation depth \(D\), use the
depth-restricted covering number \(\chi_G^D(T)\) defined above.  Then
\[
  B+M \geq \log_2 \chi_G^D(T).
\]
This expresses a communication--memory--computation trade-off: more local
computation may reduce the number of classical explanations that must be
communicated or stored, while restricted local computation forces larger
communication or memory.  Equivalently, one may use the inverse trade-off
function \(D_T(K)\) to state a lower bound on local depth for a given
coordination budget \(K\leq2^{B+M}\).  This is still a coarse trade-off: the
bound depends on the total number of available coordination states, not on a
separate lower bound forcing both \(B>0\) and \(M>0\).

\section{A Worked Global-Contextual Example}

We now analyze the \(2\times4\) polarization-path obstruction of
Ref.~\cite{yang2026global}.  The relevant data are locally noncontextual and
admit a \(\GLHV\) block description, but they cannot be glued into one
\(\GNCHV\) model.  We use Ref.~\cite{yang2026global} for the local \(\NCHV\)
and \(\GLHV\) certificates and do not rederive them here.

There are two distinct resource statements.  First, we retain the complete
probability table and apply the general separator theorem.  This gives a
lower bound against arbitrary finite classical boundary states, without a
chart-selection assumption.  Second, we discard the nonzero probabilities
and compile the Hardy support obstruction into a seven-test relation game.
That projection gives a weaker but directly contextual
cover-admissible bound.  Keeping the two levels separate is essential:
nonnegative rank measures the coordination needed to generate a probability
table across a chosen cut, whereas \(\chi_G\) measures the number of
context-independent global charts needed to cover its support constraints.

\subsection{The Hardy support structure}

There are six binary observables
\[
  X_1,\;Y_1,\;X_1',\;Y_1',\;X_2,\;Y_2.
\]
The first two act on Bob's polarization degree of freedom, the primed
observables act on Bob's path degree of freedom, and \(X_2,Y_2\) act on
Alice.  The compatible triples used in the Hardy obstruction are:
\[
\begin{array}{lll}
(X_1,Y_1',Y_2),&
(Y_1,X_1',Y_2),&
(Y_1,Y_1',X_2),\\[1mm]
(X_1,X_1',X_2).&
\end{array}
\]
The quantum construction enforces six zero-probability events:
\begin{align}
 p(x_1{=}1,\,y_1'{=}1,\,y_2{=}1)&=0, \label{eq:ourH1}\\
 p(y_1{=}1,\,x_1'{=}1,\,y_2{=}1)&=0, \label{eq:ourH2}\\
 p(y_1{=}1,\,y_1'{=}1,\,x_2{=}1)&=0, \label{eq:ourH3}\\
 p(x_1{=}1,\,y_1'{=}0,\,y_2{=}0)&=0, \label{eq:ourH4}\\
 p(y_1{=}0,\,x_1'{=}1,\,y_2{=}0)&=0, \label{eq:ourH5}\\
 p(y_1{=}0,\,y_1'{=}0,\,x_2{=}1)&=0, \label{eq:ourH6}
\end{align}
and a positive target event
\begin{equation}
  p_{\rm Q}(x_1{=}1,x_1'{=}1,x_2{=}1)
  =
  q
  =
  h_0^2(1-h_0^2)>0.
  \label{eq:ourHardyTarget}
\end{equation}

The logical contradiction is simple.  Suppose a deterministic global
assignment has
\[
  x_1=x_1'=x_2=1.
\]
If \(y_1'=0\), then the zero constraints force
\[
  y_2=1.
\]
Indeed, \(y_2=0\) would trigger Eq.~\eqref{eq:ourH4}.  Now
Eq.~\eqref{eq:ourH2}, together with \(x_1'=1\) and \(y_2=1\), forces
\(y_1=0\), whereas Eq.~\eqref{eq:ourH6}, together with \(y_1'=0\) and
\(x_2=1\), forces \(y_1=1\).  This is impossible.
If \(y_1'=1\), then the zero constraints force
\[
  y_2=0.
\]
Here \(y_2=1\) would trigger Eq.~\eqref{eq:ourH1}.  Now
Eq.~\eqref{eq:ourH3}, together with \(y_1'=1\) and \(x_2=1\), forces
\(y_1=0\), whereas Eq.~\eqref{eq:ourH5}, together with \(x_1'=1\) and
\(y_2=0\), forces \(y_1=1\).  This is impossible.
Thus any \(\GNCHV\) global assignment satisfying all six zero constraints must have
\[
  p(x_1{=}1,x_1'{=}1,x_2{=}1)=0.
\]

Equivalently, every \(\GNCHV\) model satisfies the Hardy witness inequality
\[
\begin{aligned}
 \mathcal W_H
 &=
 p(x_1{=}1,x_1'{=}1,x_2{=}1)
 -p(x_1{=}1,y_1'{=}1,y_2{=}1) \\
 &\quad
 -p(y_1{=}1,x_1'{=}1,y_2{=}1)
 -p(y_1{=}1,y_1'{=}1,x_2{=}1) \\
 &\quad
 -p(x_1{=}1,y_1'{=}0,y_2{=}0)
 -p(y_1{=}0,x_1'{=}1,y_2{=}0)\\
 &\quad
 -p(y_1{=}0,y_1'{=}0,x_2{=}1)
 \leq 0,
\end{aligned}
\]
while the quantum model has
\[
  \mathcal W_H=q>0.
\]

\subsection{The full probability table}

We first keep the complete probabilities rather than only the seven Hardy
events.  Let \(x\) be Alice's setting, \(a\) her outcome, \(C_B\) Bob's
compatible measurement block, and \(b\) its joint outcome.  For any
full-support input distributions \(\mu_A,\mu_B\), form the event matrix
\begin{equation}
 J_\mu[(x,a),(C_B,b)]
 =
 \mu_A(x)\mu_B(C_B)p(a,b\mid x,C_B).
 \label{eq:hardy-event-matrix}
\end{equation}
Multiplying rows or columns by positive numbers does not change ordinary,
nonnegative, or positive-semidefinite rank.  We may therefore remove the
input weights and analyze the corresponding \(4\times16\) matrix \(M\).

At the balanced Hardy point \(h_0^2=1/2\), relabel the four row events and
Bob's sixteen block events by \(r\in\{0,1,2,3\}\) and
\((j,k)\in\{0,1,2,3\}^2\), respectively.  With
\[
 (\phi_0,\phi_1,\phi_2,\phi_3)
 =
 \left(\frac{\pi}{3},\frac{4\pi}{3},
       \frac{5\pi}{6},\frac{11\pi}{6}\right),
\]
the table obtained from the amplitudes of Ref.~\cite{yang2026global} is
\begin{equation}
 M_{r,(j,k)}
 =
 \frac{1+\cos(\phi_r+\phi_j+\phi_k-\pi)}{8}.
 \label{eq:hardy-full-table}
\end{equation}
Thus every entry is \(0\), \(1/8\), or \(1/4\).

\begin{proposition}[Ranks of the complete Hardy table]
\label{prop:hardy-ranks}
At \(h_0^2=1/2\),
\[
 \operatorname{rank}M=3,\qquad
 \operatorname{rank}_+M=4,\qquad
 \operatorname{rank}_{\rm psd}M=2.
\]
\end{proposition}

\begin{proof}
Direct symbolic multiplication gives
\[
 \det(\lambda I-MM^{\mathsf T})
 =
 \frac{\lambda(\lambda-1)(4\lambda-1)^2}{16},
\]
so \(M\) has rank \(3\).

The four-sector \(\GLHV\) certificate of
Ref.~\cite{yang2026global} gives a nonnegative factorization with four hidden
sectors, each of weight \(1/4\), and hence
\(\operatorname{rank}_+M\leq4\).  Conversely, after reordering, four columns
contain, up to the positive factor \(1/8\), the submatrix
\begin{equation}
 N=
 \begin{pmatrix}
 0&2&1&1\\
 2&0&1&1\\
 1&1&0&2\\
 1&1&2&0
 \end{pmatrix}.
 \label{eq:hardy-neq-submatrix}
\end{equation}
Its support is the \(4\times4\) not-equal matrix.  The support of every
nonnegative rank-one summand is a rectangle avoiding the diagonal.  If three
rectangles sufficed, associate with each row the subset of rectangles
containing that row on their row side.  The four resulting subsets of a
three-element set would have to be pairwise incomparable, while Sperner's
theorem allows at most three.  Therefore at least four rectangles, and hence
at least four nonnegative rank-one summands, are necessary.  This proves
\(\operatorname{rank}_+M=4\).

The underlying state has Schmidt rank \(2\) across the Alice--Bob cut, which
gives a positive-semidefinite factorization of size \(2\).  A
positive-semidefinite factorization of size \(1\) would imply ordinary rank
\(1\), contradicting \(\operatorname{rank}M=3\).  Hence
\(\operatorname{rank}_{\rm psd}M=2\).
\end{proof}

By Theorem~\ref{thm:general-separator}, every exact general classical
simulation whose past-dependent common cause and transcript are counted
across the Alice--Bob separator obeys
\begin{equation}
 B_\Sigma+M_\Sigma\geq\log_2 4=2.
 \label{eq:hardy-one-shot-classical}
\end{equation}
The four-sector \(\GLHV\) model attains this state count.  In the standard
quantum correlation-complexity convention, the Schmidt-rank-two realization
uses a shared-seed dimension \(2\), or
\[
 Q_\Sigma=1
\]
qubit on either side of the Schmidt cut.  Thus the one-copy comparison is
four classical coordination states versus a two-dimensional quantum
separator.

This accounting must be read literally.  If unlimited shared classical
randomness is declared free, the existing \(\GLHV\) model has zero
communication cost.  Equation~\eqref{eq:hardy-one-shot-classical} instead
counts the capacity of the common classical cause as \(M_\Sigma\), or the
communication required to establish it as \(B_\Sigma\).  It is therefore a
total classical coordination cost, not a Bell communication lower bound.

\subsection{Exact and robust repetition}

For \(m\) independent copies, ordinary rank is multiplicative:
\[
 \operatorname{rank}(M^{\otimes m})=3^m.
\]
Since nonnegative rank is at least ordinary rank, whereas the product quantum
realization has positive-semidefinite rank at most \(2^m\),
\begin{equation}
 B_{\Sigma,m}+M_{\Sigma,m}
 \geq m\log_2 3,
 \qquad
 Q_{\Sigma,m}\leq m.
 \label{eq:hardy-exact-repeat}
\end{equation}
The exact classical--quantum coordination gap is therefore at least
\[
 (\log_2 3-1)m\approx0.5850\,m.
\]
No multiplicativity assumption for nonnegative rank is used.

There is also a robust asymptotic statement.  Bob's sixteen event columns
fall into four correlation types together with an independent fourfold local
duplicate.  Discarding that duplicate gives the normalized joint
distribution
\begin{equation}
 P_{UT}
 =
 \frac1{16}
 \begin{pmatrix}
 0&2&1&1\\
 2&0&1&1\\
 1&1&0&2\\
 1&1&2&0
 \end{pmatrix}.
 \label{eq:hardy-reduced-distribution}
\end{equation}
Both marginals are uniform and \(H(U,T)=7/2\) bits.

\begin{proposition}[Wyner common information of the Hardy reduction]
\label{prop:hardy-wyner}
The distribution in Eq.~\eqref{eq:hardy-reduced-distribution} satisfies
\[
 C_{\rm W}(U;T)=\frac32\ \text{bits}.
\]
\end{proposition}

\begin{proof}
For every Markov decomposition \(U-W-T\), each product component
\(p(u\mid w)p(t\mid w)\) must avoid the zero diagonal.  Its row and column
supports are therefore disjoint.  If their sizes are \(a_w,b_w\), then
\(a_w+b_w\leq4\), and
\[
 H(U\mid W=w)+H(T\mid W=w)
 \leq\log_2(a_wb_w)\leq2.
\]
Consequently
\[
 I(U,T;W)
 =
 H(U,T)-H(U\mid W)-H(T\mid W)
 \geq\frac72-2=\frac32.
\]
Equality is attained by four equiprobable values of \(W\), with uniform
conditional supports
\[
\begin{array}{c|c|c}
 w&\operatorname{supp}(U\mid w)&\operatorname{supp}(T\mid w)\\ \hline
 0&\{0,2\}&\{1,3\}\\
 1&\{0,3\}&\{1,2\}\\
 2&\{1,2\}&\{0,3\}\\
 3&\{1,3\}&\{0,2\}.
\end{array}
\]
The resulting four \(2\times2\) product rectangles reconstruct
Eq.~\eqref{eq:hardy-reduced-distribution}.
\end{proof}

Wyner's common-information theorem therefore gives, for asymptotically
vanishing total-variation error,
\begin{equation}
 B_{\Sigma,m}+M_{\Sigma,m}
 \geq\frac32m-o(m),
 \qquad
 Q_{\Sigma,m}\leq m.
 \label{eq:hardy-robust-repeat}
\end{equation}
For freely chosen inputs this is a lower bound obtained by selecting the
uniform input experiment: a simulator valid for every input must in
particular generate this distribution.  The Wyner-optimal factorization need
not itself respect measurement independence, so
Eq.~\eqref{eq:hardy-robust-repeat} is a lower bound, not a claim that the
optimal free-input classical rate is exactly \(3/2\).

Finally, the rank and common-information gaps are properties of the complete
probability table across the chosen cut.  The failure of a \(\GNCHV\) gluing
explains why this table is the genuinely global example of interest, but the
present calculation does not prove that the numerical rank gap is a monotone
of genuine global contextuality.  Noncontextual tables can also have
nontrivial correlation complexity.  Establishing a direct quantitative link
between the \(\GNCHV\) obstruction and separator complexity remains open.

\subsection{The associated seven-test game}

Define a relation game \(H_q\) with seven equally likely tests:

\begin{itemize}[leftmargin=2em]
\item six zero tests, one for each event in
Eqs.~\eqref{eq:ourH1}--\eqref{eq:ourH6}; on such a test the players win
unless the corresponding forbidden event occurs;
\item one target test on context \((X_1,X_1',X_2)\); on this test they win
iff \(x_1=x_1'=x_2=1\).
\end{itemize}

\begin{proposition}[Single-assignment bound for the Hardy game]
Every deterministic \(\GNCHV\) global assignment wins at most six of the seven
tests.  Hence
\[
  \beta_G(H_q)=\frac67.
\]
\end{proposition}

\begin{proof}
If the assignment fails the target event, it loses the target test and hence
wins at most six tests.  If it satisfies the target event
\(x_1=x_1'=x_2=1\), the Hardy implication above shows that it must trigger at
least one of the six forbidden events, and hence it loses at least one zero
test.  Thus no assignment wins all seven tests.  Since assignments satisfying
all six zero constraints and failing the target exist, the bound \(6/7\) is
tight.
\end{proof}

\begin{proposition}[Covering number of the Hardy game]
The seven-test Hardy game satisfies
\[
  \chi_G(H_q)=2.
\]
\end{proposition}

\begin{proof}
No single assignment wins all seven tests, because any assignment winning the
target test satisfies \(x_1=x_1'=x_2=1\), and the Hardy implication above shows
that it must fail at least one of the six zero tests.  Hence
\(\chi_G(H_q)\geq2\).

For the upper bound, take one assignment with
\[
  x_1=x_1'=x_2=0.
\]
The remaining three values may be chosen arbitrarily.  This assignment wins all
six zero tests, because none of the six forbidden events is triggered, but it
loses the target test.  Take a second assignment with
\[
  x_1=x_1'=x_2=1.
\]
Again choose the remaining values arbitrarily.  This assignment wins the
target test; it need not win all zero tests, because the first assignment
already covers those.  Together the two assignments cover all seven tests, so
\(\chi_G(H_q)\leq2\).
\end{proof}

Thus a perfect cover of the seven labelled tests already requires at least one
coordination bit.  The repeated-task bound below is stronger for comparing
against the quantum success probability, because it uses the single-chart
success fraction \(\beta_G(H_q)=6/7\) rather than the perfect-cover number
\(\chi_G(H_q)=2\).

The quantum strategy wins all six zero tests with probability one, and wins
the target test with probability \(q=h_0^2(1-h_0^2)\).  Under the uniform
input distribution,
\[
  \omega_Q(H_q)=\frac{6+q}{7}
  >
  \frac67
  =
  \beta_G(H_q).
\]
The maximal value of \(q\) is \(1/4\), attained at \(h_0^2=1/2\), giving
\[
  \omega_Q=\frac{25}{28}.
\]

\subsection{Support-cover coordination lower bound}

Consider \(m\) independent repetitions of the seven-test Hardy game, with the
winning condition that all copies are won.  A single global assignment wins
with probability at most
\[
  \left(\frac67\right)^m.
\]
If a cover-admissible classical protocol can coordinate among at most \(K\)
deterministic global assignments, a union bound gives success probability at most
\[
  K\left(\frac67\right)^m.
\]
To match the quantum success probability
\[
  \left(\frac{6+q}{7}\right)^m,
\]
one must have
\[
  K \geq
  \left(\frac{6+q}{6}\right)^m
  =
  (1+q/6)^m.
\]
Since such a chart-selecting protocol with \(B\) bits of communication and
\(M\) bits of coordination memory has at most \(2^{B+M}\) coordination
states, we obtain
\[
  B+M
  \geq
  m\log_2(1+q/6).
\]
At the balanced point \(q=1/4\),
\[
  B+M
  \geq
  m\log_2(25/24)
  \approx
  0.0589\,m.
\]

This is the possibilistic chart projection of the complete-table result.  Its
smaller coefficient reflects both the discarded probability information and
the restriction to global-chart selection.  Within that restricted model it
has the logical form
\[
  \NCHV\text{-local}+\GLHV+\text{not }\GNCHV
  \quad\Longrightarrow\quad
  \text{positive chart coordination cost}.
\]

\section{A Postselected KCBS Global-Contextual Task}

The Hardy task is support-based and gives a small constant.  The
\(2\times3\) KCBS construction of Ref.~\cite{yang2026global} gives a different
kind of example: after postselecting Alice's outcome \(A_1=+1\), Bob's
conditional qutrit state violates the KCBS inequality, even though the
unconditional Bob marginal is KCBS-noncontextual and the tested multipartite
scenario has a \(\GLHV\) certificate at the reported point \(c_0=1/4\).  This gives
a stronger conditional coordination rate.

Consider the five KCBS contexts
\[
  C_j=\{B_j,B_{j+1}\},
  \qquad j=0,\ldots,4
\]
with indices modulo five, and define the anti-correlation relation task
\[
  W_j=\{(b_j,b_{j+1}): b_jb_{j+1}=-1\}.
\]
This is the usual odd-cycle parity obstruction.

\begin{proposition}[KCBS anti-correlation bound]
For the five-cycle anti-correlation task,
\[
  \beta_G=\frac45,
  \qquad
  \chi_G=2.
\]
\end{proposition}

\begin{proof}
A deterministic global assignment gives values
\[
  b_0,\ldots,b_4\in\{\pm1\}.
\]
If all five anti-correlation constraints held, multiplying them would give
\[
  (b_0b_1)(b_1b_2)(b_2b_3)(b_3b_4)(b_4b_0)=-1.
\]
The left side is \(1\), since every \(b_j\) appears twice.  Hence at most four
of the five constraints can be satisfied, and \(\beta_G\le4/5\).  Alternating
values around the cycle satisfy four constraints, so \(\beta_G=4/5\).  The
same alternating assignment and its one-step shift cover all five edges, so
\(\chi_G=2\).
\end{proof}

For the conditional state \(|0\rangle\) in the KCBS construction,
Ref.~\cite{yang2026global} gives
\[
  \langle D\rangle_{A_1=+1}
  =
  5-4\sqrt5,
  \qquad
  D=\sum_{j=0}^{4}B_jB_{j+1}.
\]
By symmetry each adjacent product has expectation
\[
  \langle B_jB_{j+1}\rangle
  =
  \frac{5-4\sqrt5}{5}
  =
  1-\frac4{\sqrt5}.
\]
Therefore the conditional quantum winning probability for one KCBS
anti-correlation test is
\[
  \omega_Q^{\rm cond}
  =
  \Pr[b_jb_{j+1}=-1\mid A_1=+1]
  =
  \frac{1-\langle B_jB_{j+1}\rangle}{2}
  =
  \frac2{\sqrt5}.
\]
This exceeds the global noncontextual value \(4/5\).

For \(m\) independent postselected KCBS tests, matching the conditional
quantum success probability requires
\[
  K
  \geq
  \left(
    \frac{2/\sqrt5}{4/5}
  \right)^m
  =
  \left(\frac{\sqrt5}{2}\right)^m.
\]
Thus every cover-based classical simulation of the conditional branch
satisfies
\[
  B+M
  \geq
  m\log_2\!\left(\frac{\sqrt5}{2}\right)
  \approx
  0.1610\,m.
\]

This rate is larger than the Hardy rate \(m\log_2(25/24)\).  Its status is
different, however: it is a postselected, inequality-based lower bound.  At
the \(\GLHV\)-certified point \(c_0=1/4\), Ref.~\cite{yang2026global} reports
that the branch \(A_1=+1\) occurs with probability \(1/16\).  Thus the
statement above should not be read as an
unconditional success-probability separation.  Rather, it says that once the
globally contextual postselected branch is selected, the resulting KCBS
statistics require a larger classical chart-coordination rate.

The flagged \(3\times3\) qutrit Werner-local construction in
Ref.~\cite{yang2026global} gives a state-level version of the same conditional
mechanism.  At the reported point \(\epsilon=w=1/2\), Alice's flag outcome has
probability \(2/3\), and Bob's conditional KCBS value is
\[
  \langle D\rangle_{B|0}
  =
  \frac{35-33\sqrt5}{12}.
\]
By the same five-cycle anti-correlation conversion, the conditional winning
probability is
\[
  \omega_Q^{\rm flag}
  =
  \frac12\left(
    1-\frac1{5}\frac{35-33\sqrt5}{12}
  \right)
  \simeq 0.82325.
\]
Since the classical value is still \(4/5\), the corresponding repeated-branch
rate is
\[
  B+M
  \geq
  m\log_2\!\left(\frac{\omega_Q^{\rm flag}}{4/5}\right)
  \simeq
  0.0413\,m.
\]
This is weaker than both the \(2\times3\) postselected KCBS rate and the Hardy
rate, so we record it only as a state-level consistency check rather than as a
main quantitative example.

\section{Flag Lifts and an Exact Quadratic Separation}

The scaling discussion above says what a strong separation must look like.
We first state a general flag-lifting mechanism and then instantiate it by a
finite stabilizer interface.  The resulting theorem is no longer merely a
programme conditioned on model alignment: it applies to arbitrary exact
finite-state causal online classical simulators under the complete-boundary-
configuration convention of Sec.~\ref{sec:separator}.

The parameter controlling such an embedding is a domination weight.  For an
empirical model \(S=\{S_C\}_{C\in\C}\), define
\[
  \pi_{\NC}(S)
  =
  \max_{N\in\NC}
  \min_{C,o:\,S_C(o)>0}\frac{N_C(o)}{S_C(o)},
\]
where the maximum ranges over noncontextual empirical models \(N\) on the same
finite scenario.  Equivalently, \(\pi_{\NC}(S)\) is the largest
weight \(p\) for which there exists a noncontextual model \(N\) satisfying
\[
  N_C(o)\geq p\,S_C(o)
  \qquad
  \text{for every context }C\text{ and outcome }o.
\]
The maximum exists because the noncontextual models form a compact polytope
and the displayed objective is continuous.  Because full-support
noncontextual models exist on every finite scenario, \(\pi_{\NC}(S)>0\).  If
\(\pi_{\NC}(S_n)\) is bounded below by a constant along a seed family, then the
flag used below can also have constant probability.
For a contextual seed \(S\), one also has \(\pi_{\NC}(S)<1\); otherwise some
noncontextual \(N\) would satisfy \(N_C(o)\geq S_C(o)\) for all \(C,o\), hence
\(N=S\) by normalization, contradicting contextuality.

The following elementary construction gives the abstract form of such an
embedding.  Let \(S=\{S_C\}_{C\in\C}\) be any empirical model on a finite
measurement scenario, and let \(N=\{N_C\}_{C\in\C}\) be a full-support
noncontextual empirical model on the same scenario.  Choose
\[
  0<p<1,
  \qquad
  p\leq
  \min_{C,o:\,S_C(o)>0}\frac{N_C(o)}{S_C(o)}.
\]
Then
\[
  R_C(o)=\frac{N_C(o)-pS_C(o)}{1-p}
\]
is again an empirical model: nonnegativity follows from the choice of \(p\),
normalization is immediate, and no-disturbance is preserved because \(R\) is
an affine combination of the no-disturbing models \(N\) and \(S\).  Introduce
a binary flag measurement \(A\) and
define the extended model \(G^{p,N}(S)\) on contexts \(\{A\}\cup C\) by
\[
  G^{p,N}_{A,C}(1,o)=pS_C(o),
  \qquad
  G^{p,N}_{A,C}(0,o)=(1-p)R_C(o).
\]
The unconditioned marginal on the original measurements is \(N\):
\[
  \sum_{a\in\{0,1\}}G^{p,N}_{A,C}(a,o)=N_C(o).
\]
Thus the original subsystem is locally noncontextual after the flag is
ignored, while conditioning on \(A=1\) recovers the seed model \(S\).

\begin{proposition}[Hidden-KS lift]
If the seed model \(S\) is KS-contextual, then the extended model
\(G^{p,N}(S)\) is globally KS-contextual.  Nevertheless, the marginal on the
original subsystem is the noncontextual model \(N\), and the flag subsystem is
classical.
\end{proposition}

\begin{proof}
The marginal claim follows from the displayed identity.  Suppose, toward a
contradiction, that \(G^{p,N}(S)\) admitted a global noncontextual hidden
variable distribution \(\mu\) over the flag value \(a\) and all original
measurement outcomes.  Since \(p>0\), conditioning \(\mu\) on \(a=1\) gives a
well-defined global distribution over the original measurements.  For every
context \(C\), its marginal is exactly
\[
  G^{p,N}_{A,C}(o\mid A=1)=S_C(o).
\]
This is a global noncontextual model for \(S\), contradicting the
KS-contextuality of the seed.
\end{proof}

This lift should be read as a clean mathematical mechanism, not as a complete
physical construction by itself.  At the level of a single tested block
\(\{A\}\cup C\), the model is just an ordinary joint distribution and hence
has a blockwise classical explanation.  What fails is the gluing of those
blockwise explanations into one context-independent global chart.  To obtain
the exact ``local \(\NCHV\) plus \(\GLHV\)-local blocks'' property in a concrete
multipartite experiment, one must realize this hiding mechanism inside a
physical scenario such as the polarization-path or postselected constructions
of Ref.~\cite{yang2026global}.

There is also a support-versus-distribution distinction.  The lifted model's
unconditioned support may be easy to cover, because the \(A=0\) filler branch
can contain many classical outcomes.  The lift is therefore a strong
tool for exact probabilistic simulation, approximate simulation with controlled
conditioning, or explicitly flagged tasks.  It should not be interpreted as
automatically producing a large bare support-covering number such as
\[
 \chi_G(G^{p,N}(S)).
\]

We isolate the required reduction in a form that can be checked independently
of any particular construction.  Let \(S\) be a KS-contextual seed task and
let \(G\) be a larger multipartite task.  We say that \(S\) is a
\emph{flag reduction} of \(G\) if there is an event \(F\) of positive
probability and a relabelling of the contexts and outcomes in the conditional
model \(G\mid F\) such that the relabelled conditional model is exactly \(S\).
The reduction is \emph{cost-preserving} for a class of classical simulators if
postselecting a simulator for \(G\) on \(F\), and then applying the relabelling,
produces a simulator for \(S\) with no larger set of internal classical states
or global charts.

\begin{proposition}[Cost inheritance under flag reductions]
Let \(S\) be a cost-preserving flag reduction of \(G\).  If \(K_{\mathcal P}(T)\)
denotes the minimum number of classical states, charts, or coordination labels
needed by a simulator class \(\mathcal P\) to simulate a task \(T\) exactly, then
\[
  K_{\mathcal P}(G)\geq K_{\mathcal P}(S).
\]
Equivalently,
\[
  \log_2 K_{\mathcal P}(G)\geq \log_2 K_{\mathcal P}(S).
\]
\end{proposition}

\begin{proof}
Take an optimal \(\mathcal P\)-simulator for \(G\).  Condition its classical
state distribution and output rule on the flag event \(F\), and then apply the
fixed relabelling from \(G\mid F\) to \(S\).  By cost preservation this is a
valid \(\mathcal P\)-simulator for \(S\) using no more classical states than
the original simulator for \(G\).  Therefore the minimum cost for \(G\) cannot
be smaller than the minimum cost for \(S\).
\end{proof}

The same reduction has an approximate version.  Suppose the target model \(G\)
has flag probability \(p_F>0\), independent of the post-flag context, and a
classical simulator \(\widehat G\) satisfies
\[
  \TV(\widehat G_C,G_C)\leq \epsilon
\]
for every flagged context \(C\).  If \(\epsilon<p_F\), then the simulator's
conditional distribution given \(F\) is well-defined and satisfies
\[
  \TV(\widehat G_C(\,\cdot\,|F),G_C(\,\cdot\,|F))
  \leq
  \frac{2\epsilon}{p_F}.
\]
Indeed, for the event \(F\), the subnormalized measures
\(\widehat G_C(\cdot,F)\) and \(G_C(\cdot,F)\) differ by at most \(\epsilon\)
in total variation, and their total masses differ by at most \(\epsilon\).
Normalizing by \(p_F\) therefore costs at most another \(\epsilon/p_F\).  Thus
an \(\epsilon\)-accurate simulator for the lifted model gives a
\((2\epsilon/p_F)\)-accurate simulator for the seed.

Consequently, if simulating \(S_n\) within error \(\delta\) requires
\(L(n,\delta)\) bits of memory or coordination, then simulating a lifted model
\(G_n\) within error \(\epsilon\) requires at least
\[
  L\!\left(n,\frac{2\epsilon}{p_{F,n}}\right)
\]
bits, as long as the flag reduction is cost-preserving.  Constant flag
probability preserves approximate lower bounds up to constant factors.

Putting these observations together gives the strong-separation transfer
principle.

\begin{proposition}[Strong-separation transfer]
Let \(\{S_n\}\) be a seed family whose \(\delta\)-accurate classical simulation
within a simulator class \(\mathcal P\) requires at least \(L(n,\delta)\) bits
of memory or coordination.  Suppose \(\{G_n\}\) is a family of genuine
global-KS lifts of \(\{S_n\}\) with flag probabilities
\[
  p_{F,n}\geq p_0>0,
\]
and suppose the flag reductions are cost-preserving for \(\mathcal P\).  Then
any \(\epsilon\)-accurate \(\mathcal P\)-simulation of \(G_n\), with
\(\epsilon<p_0/2\), requires at least
\[
  L\!\left(n,\frac{2\epsilon}{p_0}\right)
\]
bits of memory or coordination.  In particular, if
\[
  L(n,\delta_0)=\Omega(f(n))
\]
for some constant accuracy \(\delta_0>0\), then every simulation of \(G_n\)
with \(\epsilon\leq p_0\delta_0/2\) has cost \(\Omega(f(n))\).
\end{proposition}

\begin{proof}
An \(\epsilon\)-accurate simulator for \(G_n\) gives, after conditioning on
the flag event, a \((2\epsilon/p_{F,n})\)-accurate simulator for \(S_n\) in the
same simulator class and using no more memory or coordination states.  Since
\(p_{F,n}\geq p_0\), this error is at most \(2\epsilon/p_0\).  The claimed
lower bound is exactly the seed lower bound applied to the induced simulator
for \(S_n\).
\end{proof}

\begin{theorem}[Abstract strong genuine-global-KS lift]
Let \(\{S_n\}\) be any finite KS-contextual seed family, and let
\(\mathcal P\) be a classical simulator class closed under conditioning on
positive-probability flag events.  For each \(n\), choose
\[
  0<p_n<1,
  \qquad
  p_n\leq \pi_{\NC}(S_n)
\]
and a noncontextual model \(N_n\) witnessing this domination.  Then the
explicit lifted family
\[
  G_n=G^{p_n,N_n}(S_n)
\]
has the following properties:
\[
\begin{array}{ll}
\text{(i)} & \text{the original subsystem marginal is noncontextual;}\\
\text{(ii)}& \text{the flag subsystem is classical and each flagged block has a joint model;}\\
\text{(iii)}& \text{the full model is globally KS-contextual;}\\
\text{(iv)}& K_{\mathcal P}(G_n)\geq K_{\mathcal P}(S_n)
              \text{ for exact simulation.}
\end{array}
\]
If, in addition, \(p_n\geq p_0>0\) and \(\delta\)-accurate simulation of
\(S_n\) requires \(L(n,\delta)\) bits of memory or coordination, then
\(\epsilon\)-accurate simulation of \(G_n\) requires at least
\[
  L\!\left(n,\frac{2\epsilon}{p_0}\right)
\]
bits for every \(\epsilon<p_0/2\).
\end{theorem}

\begin{proof}
Items (i)--(iii) are exactly the hidden-KS lift proposition applied to
\(S_n\), \(N_n\), and \(p_n\).  Item (iv) is the cost-inheritance proposition,
because conditioning \(G_n\) on the flag event \(A=1\) recovers \(S_n\).  The
approximate statement is the strong-separation transfer proposition.
\end{proof}

The abstract theorem isolates the logic of the reduction.  We instantiate its
cost-transfer part using the published stabilizer support-overlap bound.  A
Peres--Mermin measurement branch supplies a noncontextual hiding model with
constant flag probability, and conditioning preserves the state capacity of
every exact causal online simulator.

\subsection{A finite stabilizer witness}

Let
\[
 N_n
 =
 2^n\prod_{j=1}^{n}(2^j+1)
\]
be the number of pure \(n\)-qubit stabilizer states.  The counting argument of
Karanjai, Wallman, and Bartlett shows that a classical ontic state compatible
with exact stabilizer updates and single-shot distinguishability can occur in
the supports of at most
\[
 m_n=5\,3^{n-2}
\]
pure stabilizer preparations~\cite{karanjai2018contextuality}.

\begin{lemma}[Finite causal witness from the stabilizer overlap bound]
\label{lem:finite-stabilizer-witness}
For every \(n\geq2\), there is a finite adaptive stabilizer interface \(W_n\)
whose exact behavior \(S_{W_n}\) requires at least
\begin{equation}
 K_n
 \geq
 \frac{N_n}{m_n}
 =
 \frac{2^n\prod_{j=1}^{n}(2^j+1)}
      {5\,3^{n-2}}
 \label{eq:finite-stabilizer-count}
\end{equation}
states in every classical causal online realization.
\end{lemma}

\begin{proof}
Use as input the Karanjai--Wallman--Bartlett theorem that every set of more
than \(m_n\) pure stabilizer states admits a stabilizer partitioning
measurement~\cite{karanjai2018contextuality}.  For each outcome of such a
measurement, two states in the set are mapped to orthogonal postmeasurement
records.  The sets of stabilizer states, Pauli measurements, and Clifford
updates are finite for fixed \(n\).  For each of the finitely many subsets of
\(m_n+1\) preparations, include one partitioning measurement and, at each
relevant outcome, one subsequent single-shot test that distinguishes the
resulting orthogonal pair.  Their finite union defines \(W_n\).

Suppose one boundary state \(\lambda\) occurred with positive probability
after every preparation in one such subset.  Its response to the common
partitioning measurement is preparation independent.  Choose an outcome
\(a\) with positive response probability and a state \(\lambda'\) with
positive transition probability
\(\Gamma(\lambda',a\mid\lambda)\).  By the partitioning property, two of the
preparations lead on outcome \(a\) to orthogonal records.  Yet the same
\(\lambda'\) occurs after both records, so the subsequent distinguishing test
would have to produce two different deterministic outcomes from the same
future-accessible state and the same query.  This contradicts exact causal
simulation.  Hence no boundary state can occur in more than \(m_n\)
preparation supports.  Covering all \(N_n\) preparations requires at least
\(N_n/m_n\) states.
\end{proof}

The only nonanalytic base input here is the exhaustive two-qubit result
reported in Ref.~\cite{karanjai2018contextuality}: every set of six pure
two-qubit stabilizer states admits a partitioning measurement.  The present
argument uses that published result rather than claiming a new independent
enumeration certificate.

To make KS contextuality explicit, add a disjoint root branch carrying the
standard Peres--Mermin square on the first two qubits, and call the resulting
finite interface \(E_n\).  No operational events are identified between the
two root branches.  The same classical state space must serve both branches,
so restricting a simulator of \(E_n\) to \(W_n\) preserves
Eq.~\eqref{eq:finite-stabilizer-count}.

\subsection{Constant-visibility noncontextual hiding}

Write the nine Peres--Mermin outcomes as bits \(z_{ij}\).  The target parities
of the three rows and three columns are
\[
 (0,0,0,0,0,1).
\]
Let \(S_{\rm PM}\) be uniform over the four triples satisfying the target
parity in each context, and let \(R_{\rm PM}\) be uniform over all eight
triples.

\begin{lemma}[Exact Peres--Mermin visibility]
\label{lem:pm-visibility}
The mixture
\[
 (1-p)R_{\rm PM}+pS_{\rm PM}
\]
is noncontextual if and only if \(0\leq p\leq2/3\).
\end{lemma}

\begin{proof}
Every global assignment violates an odd number of the six parity equations
and therefore satisfies at most five.  The context-averaged success
probability of the mixture is \((1+p)/2\).  Noncontextuality requires
\((1+p)/2\leq5/6\), so \(p\leq2/3\).

For sufficiency, mix uniformly over the global assignments that violate
exactly one context.  Each of the six one-violation syndromes has \(16\)
solutions, giving \(96\) assignments.  For a fixed context, five syndromes
satisfy the target parity and one violates it; within each syndrome the four
triples of the required parity occur uniformly.  Each target-parity triple
therefore has probability \(5/24\), and each opposite-parity triple has
probability \(1/24\).  These are exactly the probabilities of
\[
 \frac13R_{\rm PM}+\frac23S_{\rm PM}.
\]
Smaller \(p\) follows by further mixing with \(R_{\rm PM}\).
\end{proof}

Let \(S_n=(S_{W_n},S_{\rm PM})\) be the seed behavior on the two disjoint
root branches.  Let \(R_{W_n}\) be a fixed memoryless causal reset behavior on
\(W_n\), and set
\[
 R_n=(R_{W_n},R_{\rm PM}),\qquad p=\frac23.
\]
On the Peres--Mermin branch, Lemma~\ref{lem:pm-visibility} gives the
noncontextual empirical model
\begin{equation}
 N_{\rm PM}^{\rm mix}
 =
 \frac13R_{\rm PM}+\frac23S_{\rm PM}
 \in\NC({\rm PM}).
 \label{eq:explicit-hidden-mixture}
\end{equation}
No measurement-noncontextuality claim is made for the adaptive \(W_n\)
branch.  It is a causal-process benchmark carrying the state-complexity
lower bound, whereas the Peres--Mermin branch supplies the empirical
genuine-global-KS restriction.  Arbitrary causal-policy decomposability and
measurement noncontextuality are not identified.

Add a binary flag \(F\) and define the lifted process
\begin{equation}
\begin{split}
 G_n(F=1,\tau)&=\frac23S_n(\tau),\\
 G_n(F=0,\tau)&=\frac13R_n(\tau),
\end{split}
\label{eq:finite-flag-lift}
\end{equation}
where \(\tau\) is the complete transcript on \(E_n\).

\begin{proposition}[Genuine-global properties of the finite lift]
\label{prop:finite-lift-properties}
Restricted to the flag and Peres--Mermin measurement branch, the flag
subsystem is noncontextual, ignoring the flag gives the noncontextual marginal
\(N_{\rm PM}^{\rm mix}\), the flag-versus-seed-block partition has a
\(\GLHV\) model, and the full restricted empirical model has no \(\GNCHV\)
model.
\end{proposition}

\begin{proof}
The first two claims are immediate from
Eqs.~\eqref{eq:explicit-hidden-mixture}--\eqref{eq:finite-flag-lift}.  For the
\(\GLHV\) claim, use \(\lambda\in\{0,1\}\) with probabilities \(1/3,2/3\);
the flag reports \(\lambda\), while the Peres--Mermin block uses
\(R_{\rm PM}\) for \(\lambda=0\) and \(S_{\rm PM}\) for \(\lambda=1\).
This is local across the stated partition even though the response internal
to the Peres--Mermin block is contextual for \(\lambda=1\), which is allowed
by the definition of \(\GLHV\).

If the full model admitted a \(\GNCHV\) distribution, conditioning that
distribution on the context-independent positive-probability event \(F=1\)
would give a noncontextual model of \(S_{\rm PM}\), a contradiction.
\end{proof}

The full two-branch process also has an explicit quantum realization.  For
each seed preparation record \(r\), on
\[
 \mathcal H_G=\mathcal H_{\rm stab}\oplus\mathbb C,
 \qquad
 \dim\mathcal H_{\rm stab}=2^n,
\]
prepare
\[
 \rho^G_r
 =
 \frac23\rho^S_r\oplus
 \frac13|{\perp}\rangle\langle{\perp}|
\]
and use block-diagonal instruments
\[
 \mathcal I^G_{x,a}
 =
 \mathcal I^S_{x,a}\oplus q_x(a)\operatorname{id}_{\mathbb C}.
\]
Here \(q_x(a)\) is the memoryless response chosen for the reset branch, and
\(\sum_a\mathcal I^G_{x,a}\) is trace preserving.  The same formula applies
to each allowed seed instrument, so the notation does not identify distinct
preparation or query procedures.
The projectors onto the two summands are the flag measurement.  The boundary
dimension \(2^n+1\) embeds in \(n+1\) qubits.  A manifestly bipartite
realization adds a flag qubit on Alice and keeps this direct-sum system on
Bob, using at most \(n+2\) qubits in total.

\subsection{Exact quadratic coordination theorem}

\begin{theorem}[Exact quadratic separation with a genuine-global restriction]
\label{thm:exact-quadratic}
For every \(n\geq2\), the finite benchmark \(G_n\) has an empirical
measurement restriction satisfying local \(\NCHV\), blockwise \(\GLHV\), and
global non-\(\GNCHV\).  At the separator \(\Sigma_n\) after the flag and a
stabilizer preparation but before the first adaptive partitioning query,
every exact finite-state classical causal online simulator obeys
\begin{equation}
\begin{split}
 B_{\Sigma_n}+M_{\Sigma_n}
 &\geq
 \log_2
 \frac{2^n\prod_{j=1}^{n}(2^j+1)}
      {5\,3^{n-2}}\\
 &=
 \frac12n^2+
 \left(\frac32-\log_2 3\right)n+O(1),
\end{split}
\label{eq:exact-quadratic-bound}
\end{equation}
whereas the quantum boundary satisfies
\[
 Q_{\Sigma_n}\leq n+1.
\]
\end{theorem}

\begin{proof}
Conditioning an exact classical simulator of \(G_n\) on \(F=1\) only
reweights its distribution over the existing boundary-state set
\(\Lambda\); it creates no new states.  With the public history \(F=1\)
fixed, the same response and update kernels reproduce \(S_n\).  Restricting
further to \(W_n\) and applying
Lemma~\ref{lem:finite-stabilizer-witness} gives
\[
 |\Lambda|
 \geq
 \frac{2^n\prod_{j=1}^{n}(2^j+1)}
      {5\,3^{n-2}}.
\]
All past-dependent records available after \(\Sigma_n\) are included in
\(\Lambda\), so
\(|\Lambda|\leq2^{B_{\Sigma_n}+M_{\Sigma_n}}\).  Taking logarithms proves the
first line.  Expanding
\[
 \log_2(2^j+1)=j+\log_2(1+2^{-j})
\]
gives the stated asymptotic formula.  The quantum upper bound is the
direct-sum realization above, and
Proposition~\ref{prop:finite-lift-properties} supplies the genuine-global
properties of the empirical flag--Peres--Mermin restriction.
\end{proof}

The lower bound permits unrestricted local computation depth.  Computation
can transform information already on the future side of \(\Sigma_n\), but it
cannot distinguish more past preparations than the counted boundary
configuration distinguishes.  Recomputing from a saved preparation label or
transcript merely places that information in \(B_{\Sigma_n}\) or
\(M_{\Sigma_n}\).

Equivalently,
\[
 \operatorname{rank}_+^{\rm causal}(H_n)
 \geq
 \frac{2^n\prod_{j=1}^{n}(2^j+1)}
      {5\,3^{n-2}}
 =
 2^{\Omega(n^2)},
\qquad
 \operatorname{rank}_{\rm psd}^{\rm causal}(H_n)\leq2^n,
\]
for the process Hankel table \(H_n\) of the stabilizer branch.  The
shift-consistent update requirement is essential: this does not assert that
the ordinary nonnegative rank of every static flattening is
\(2^{\Omega(n^2)}\).

At \(n=2\), \(N_2=60\) and \(m_2=5\), so
\[
 |\Lambda|\geq12,\qquad
 B_{\Sigma_2}+M_{\Sigma_2}\geq\log_2 12.
\]
Thus an integral classical budget needs at least four bits, while the
post-flag quantum boundary fits in three qubits.

The theorem is exact.  Its flag visibility is the size-independent constant
\(\Pr(F=1)=2/3\), so conditioning an \(\epsilon\)-approximate simulator
amplifies total-variation error by at most a factor \(3\).  A
constant-error quadratic theorem would therefore follow from a corresponding
constant-error causal-state lower bound for the stabilizer witness.  That
robust seed lower bound, and a flag-free intrinsic construction, remain open.
The theorem also concerns causal state complexity rather than the bare
support-covering number: it does not imply
\(\chi_G(G_n)=2^{\Omega(n^2)}\).

\section{Discussion}

The framework separates three questions that were conflated in the initial
covering picture.  First, how much past-dependent classical information must
cross a spacetime separator to reproduce a probability table?  This is the
general \(B_\Sigma+M_\Sigma\) question, controlled one-shot by nonnegative
rank, asymptotically by common information, and sequentially by causal
positive-realization rank.  Second, how many context-independent global
charts are needed if the classical model is required to select such charts?
This is the restricted \(\chi_G^D\) question.  Third, which part of either
cost is specifically forced by genuine global KS contextuality rather than
by ordinary correlation complexity?  The examples answer the first two
questions and partially localize the third, but they do not yet define a
general genuine-global contextuality monotone.

This distinction also clarifies the role of communication, memory, and
computation.  \(B_\Sigma\) and \(M_\Sigma\) are two physical locations of the
same classical boundary information: one crosses the cut in a message and
the other persists in a state.  \(D\) is not another information unit.  It
restricts which response and update maps can be applied to the available
boundary state.  The exact stabilizer theorem remains quadratic even when
\(D\) is unrestricted, whereas the twisted-hypercube example shows an
explicit depth tradeoff inside a restricted chart class.

\subsection*{Scope and limitations}

The present results do not constitute a complete complexity theory of global
contextuality.  Several limitations remain.

First, the operational theorems mainly give cut-sum bounds such as
\[
  B_\Sigma+M_\Sigma
  \geq
  \log_2\operatorname{rank}_+F_\Sigma(P)
\]
and therefore do not force both \(B_\Sigma>0\) and \(M_\Sigma>0\)
simultaneously.  They permit a message--memory tradeoff across the chosen cut.
Rectangle selectors begin to resolve the spatial geometry of communication,
but a nontrivial \(B\)-versus-\(M\) frontier has not been computed for a
natural genuine-global family.  Wyner common information supplies an
amortized information-rate lower bound, but a full interactive information
complexity theory is not developed here.

Second, nonnegative rank is a property of a probability table and a chosen
flattening, not a contextuality monotone.  The Hardy rank gap is therefore an
operational coordination advantage exhibited by the genuinely global table,
but the present proof does not show that the gap vanishes on all \(\GNCHV\)
models.  The exact quadratic theorem uses genuine global contextuality to
hide and transfer a stabilizer state-complexity obstruction; its lower-bound
mechanism is inherited from the KS seed.  A direct intrinsic lower bound in
terms of distance from the \(\GNCHV\) set remains open.

Third, the cover-admissible results are deliberately restricted.  In that
model a transcript-memory pair determines a context-independent global
chart.  The general separator and causal-realization theorems do not impose
this restriction and allow randomized, context-dependent, adaptive response
and update kernels.  Consequently \(\chi_G^D\), nonnegative rank, and causal
realization rank should not be interchanged without a proved reduction.

Fourth, the computation-depth axis depends on a chosen local computation model.
The twisted hypercube calculation makes this dependence explicit by using a
fan-in-two parity-tree model, where depth \(D\) can combine at most \(2^D\)
input coordinates.  This gives an exact and useful toy calculation, but it is
not a physical global-KS example.  The analogous calculation for the Hardy
support task, or for a stronger \(\GNCHV\) obstruction, remains open.

Fifth, robustness is uneven.  The complete Hardy table gives a
vanishing-error amortized rate \(3/2\) classical bits per copy against one
quantum separator qubit, while its seven-test support projection has only the
smaller chart rate \(\log_2(25/24)\).  The KCBS comparison is postselected.
The stabilizer family gives an exact quadratic separation and has constant
flag visibility, but extending its \(\Omega(n^2)\) lower bound to fixed
nonzero error requires a robust causal-state theorem for the stabilizer seed.
The flag-free intrinsic genuine-global strong-separation problem also remains
open.

\subsection*{Outlook: quantum networks}

The same resource accounting is natural for quantum networks and quantum
internet scenarios, where remote nodes are connected by quantum channels,
entanglement distribution, local operations, and classical control
messages~\cite{kimble2008internet,wehner2018internet}.  A network task has
local node views, edge or small-block tests, and a desired global correlation.
In the notation of this paper one may write a network support task as
\[
  T_{\rm net}=(\C_{\rm node}\cup\C_{\rm edge}\cup\C_{\rm block},\{W_C\}),
\]
and a classical network simulator with communication \(B_{\rm net}\), node
memory \(M_{\rm node}\), and local processing depth \(D\) obeys the general
separator constraint
\[
  B_{\rm net}+M_{\rm node}
  \geq
  \log_2\operatorname{rank}_+F_\Sigma(P_{\rm net}).
\]
If the simulator additionally selects context-independent network charts,
the support projection gives
\[
 B_{\rm net}+M_{\rm node}\geq\log_2\chi_G^D(T_{\rm net}).
\]
Genuine global KS
contextuality is then a network-level obstruction: every local node and every
tested block may have a classical explanation, while the whole network lacks a
single globally consistent hidden-state chart.  This is close in spirit to
the broader study of nonclassicality in networks, including bilocal and
network-locality constraints~\cite{branciard2010bilocal}, but the present
focus is the KS-style global-chart obstruction and its coordination cost.

\section*{Acknowledgements}

The author acknowledges the use of AI-assisted tools during the preparation of
this manuscript for literature exploration, mathematical checking, language
polishing, and organizational support.  The author is solely responsible for
all mathematical claims, interpretations, and conclusions.

\end{document}